\documentclass[]{interact}

\theoremstyle{plain}
\newtheorem{Theorem}{Theorem}[section]
\newtheorem{Lemma}[Theorem]{Lemma}
\newtheorem{Corollary}[Theorem]{Corollary}

\theoremstyle{definition}

\theoremstyle{remark}
\newtheorem{Remark}{Remark}

\newcommand{\bdiag}{{\rm bdiag}}
\newcommand{\diag}{{\rm diag}}

\usepackage{bm}
\usepackage{optidef}
\usepackage[ruled,vlined]{algorithm2e}
\usepackage{hyperref}

\begin{document}
\title{Adaptive Regularized Zero-Forcing Beamforming \\ in Massive MIMO with Multi-Antenna Users}

 \author{
\name{Evgeny~Bobrov\textsuperscript{a,b,*}\thanks{Emails: eugenbobrov@ya.ru, roborisor@gmail.com, kuznetsov.victor@huawei.com, luhao12@huawei.com, minenkov.ds@gmail.com, stroshin@hse.ru, d.yudakov43@gmail.com, zaev.da@gmail.com.}, Boris~Chinyaev\textsuperscript{a,b}, Viktor~Kuznetsov\textsuperscript{b}, Hao~Lu\textsuperscript{b}, Dmitrii~Minenkov\textsuperscript{a,d}, Sergey~Troshin\textsuperscript{c}, Daniil~Yudakov\textsuperscript{a,b},  and Danila~Zaev\textsuperscript{b}}
\affil{\textsuperscript{a} M.V. Lomonosov Moscow State University, Russia;  \\
\textsuperscript{b} %
Huawei Technologies, Russian Research Institute, Moscow Research Center, Russia, \\
\textsuperscript{c} National Research University Higher School of Economics, Russia \\
\textsuperscript{d} A. Ishlinsky Institute for Problems in Mechanics RAS, Russia \\
\textsuperscript{*} Corresponding author 
}}

\maketitle

\begin{abstract}
Modern wireless cellular networks use massive multiple-input multiple-output (MIMO) technology. This technology involves operations with an antenna array at a base station that simultaneously serves multiple mobile devices which also use multiple antennas on their side. For this, various precoding and detection techniques are used, allowing each user to receive the signal intended for him from the base station. There is an important class of linear precoding called Regularized Zero-Forcing (RZF). In this work, we propose Adaptive RZF (ARZF) with a special kind of regularization matrix with different coefficients for each layer of multi-antenna users. These regularization coefficients are defined by explicit formulas based on Singular Value Decomposition (SVD) of user channel matrices. We study the optimization problem, which is solved by the proposed algorithm, with the connection to other possible problem statements. We prove theoretical estimates of the number of conditionality of the inverse covariance matrix of the ARZF method and the standard RZF method, which is important for systems with fixed computational accuracy. Finally, We compare the proposed algorithm with state-of-the-art linear precoding algorithms on simulations with the Quadriga channel model. The proposed approach provides a significant increase in quality with the same computation time as in the reference methods.
\end{abstract}

\begin{keywords}
Massive MIMO, Telecom, Optimization, MU-Precoding, Zero Forcing, SVD
\end{keywords}

\section{Introduction}
In multiple-input multiple-output (MIMO) systems with numerous antennas, precoding is an important part of downlink signal processing, since this procedure can focus the transmission signal energy on smaller areas and allows for greater spectral efficiency with less transmitted power~\cite{EE, 5G}. Various linear precodings allow directing the maximum amount of energy to the user as Maximum Ratio Transmission (MRT) or completely get rid of inter-user interference as Zero-Forcing (ZF)~\cite{ZF_MRT}. 
In the case of Regularized Zero-Forcing (RZF) algorithms (aka the Wiener filter), we balance between maximizing the signal power and minimizing the interference leakage~\cite{Joham_RZF,Bjornson,RZF19,RZF2,OptimalRegularization} and still have low complexity compared to the non-linear precoding. Note that usually a ``scalar'' regularization is considered, i.e. regularization with a scalar multiplied by a unitary matrix of the corresponding size.  E.~Bj\"ornson in ~\cite{Bjornson} uses the primal-dual approach and proves that the optimal regularization has the form of a diagonal matrix (generally speaking, with different elements). This proof is not constructive and does not provide any particular formula or algorithm for this optimal regularization. In the current paper, we propose an explicit heuristic formula for a diagonal regularization that provides better results compared with scalar RZF. We support the proposal with theoretical justification and tests using Quadriga~\cite{Quadriga}. 

There are also approaches for linear precoding, like Transceivers (detection-aware precoding, see e.g.~\cite{PrecodingDetection}), Block-diagonal precoding (see e.g.~\cite{BlockDiagonal}), and Power Allocation problem (see the involved study in~\cite{Bjornson_tb_17}).
There are also different non-linear precoding techniques such as Dirty Paper Coding (DPC) and Vector Perturbation (VP) but they have much higher implementation complexity~\cite{DPC} and in the case of usage of numerous antennas in massive MIMO, linear precoding techniques are more preferable. 
There also are many good surveys of precoding techniques~\cite{Survey2017,Survey2015, EZF19} and papers that consider different aspects and variations of these algorithms (see e.g.~\cite{Bogale, Inverses, RZF_mBS}). In particular, the work~\cite{RZF_mBS} presents various variants of RZF in the case of multiple base stations, which help to reduce inter-cell interference. 

Most works do not pay much attention to multi-antenna user equipment (UE) for simplicity. 
In our work, we consider a single base station serving multi-antenna users who simultaneously receive fewer data channels than their number of antennas. This approach is necessary because, in practice, the channels between different antennas of one UE are often spatial correlated~\cite{SpatialCorrelation}. Therefore, the matrix of the user channel is ill-conditioned (or even has incomplete rank), thus one can not efficiently transmit data using the maximum number of streams.
To solve this problem, instead of the full matrix of the user channel, vectors from its singular value decomposition (SVD) with the largest singular values are used for precoding~\cite{SVD}. In the case of UE with one antenna, the channel matrix can be normalized~\cite{RZF19} and normalization coefficients are the path losses of each UE that can differ by several orders, common Reference Signal Received Power (RSRP) values vary from $(-130) \, dBm$ to $ (-70) \, dBm$. When we use SVD, the singular values of different UEs have the order of the corresponding path-losses and therefore vary greatly. 

We managed to find a simple heuristic formula that provides gain over known RZF algorithms. Motivated by the above issue, we propose Adaptive RZF (ARZF) precoding with diagonal regularization. Algorithms of the Wiener filter type~\cite{Joham_RZF} have a scalar regularization of the form $\lambda \bm I$, and ARZF belongs to the narrow class of precoding which uses a diagonal regularization matrix. Such algorithms effectively use different regularization parameters for different UEs by taking into account the singular values of transmitted layers (streams).



The idea of ARZF is well-known, for example, in~\cite{Bjornson} it is shown that maximization of the SINR function (including the considered sum SE~\eqref{J_SE}) is achieved by an algorithm with appropriate diagonal regularization, and in~\cite{nguyen2014mmse} the authors derive a similar formula for the single-antenna MU-SIMO user system. 

The results of this paper are the adaptation of the formula $\bm W_{ARZF}(\bm V)$ together with Theorem~\ref{theorem_ARZF} for Multi-User (MU)-MIMO systems.

The rest of this paper is organized as follows. In Section~\ref{sec_model} we formulate the problem, which consists of the Channel and System Model, quality measures, and power constraints. The downlink MIMO channel model is simplified using SVD-decomposition of the channel (sec.~\ref{sec_SVD}) and useful idealistic detection (sec.~\ref{sec_ConjDet}); quality measures and power constraints are discussed in sec.~\ref{sec_QM},~\ref{sec_PC}. In Section~\ref{sec_precoding} we propose an adaptive precoding algorithm that utilizes UE singular values. Then, we study its relation with known precoding, including MRT, ZF, and RZF. Comparison of these algorithms on numerical experiments with Quadriga is provided in Section~\ref{sec_simulation}. Conclusions are drawn in Section~\ref{sec_conclusion}. Symbols and notations are shown in Tab.~\ref{table_example}.

Throughout the paper, we assume that the transmitter has perfect channel state information (CSI) of all downlink channels, this assumption is reasonable in time division duplex (TDD) systems, which allows the transmitter to employ reciprocity to estimate the downlink channels, and that each user only has access to their own CSI, but not the CSI of the downlink channels of the other users.

Throughout the paper we use the following notations. We consider one cell with $K$ UE, the number of transmit antennas is $T$ and the number of receive antennas and transmit symbols of UE $k$ are $R_k = 1,2,4$ and $L_k \leqslant R_k$ with total $R = \sum_{k=1}^K R_k$ and $L = \sum_{k=1}^K L_k$: $L_k \leqslant R_k \leqslant T$. By bold lower case letters we denote vectors: either columns or rows, which will be clear from the context. We denote matrices by bold upper case letters, considering them as sets of vectors, e.g. channel matrix is a set of vector-rows $\bm H = [\bm h_{1};...;\bm h_{R}] \in \mathbb C^{R\times T}$ and precoding matrix is a set of vector-columns $\bm W = (\bm w_1, ..., \bm w_L) \in \mathbb C^{T\times L}$. Matrix elements are denoted by ordinary lower case letters with the first index standing for rows and the second one -- for columns: $\bm H = \{h_{rt}\}, \; \bm W = \{w_{tl}\}, \; r=1,...,R, \; t = 1,...,T, \; l = 1,...,L$. Hermitian conjugate is denoted by $\bm H^ \mathrm H := \overline{\bm H}^ \mathrm T$. Diagonal and block-diagonal matrices are written as $\bm S_k = \diag\{s_{k,1},\dots,s_{k,R_k}\}$ and $\bm S = \bdiag\{\bm S_1,\dots,\bm S_K\}$ correspondingly, the identity matrix of size $T$ is $\bm I_T = \diag\{1,\dots, 1\} \in \mathbb C^{T\times T}$. Trace of a square matrix $\bm A$ is denoted by ${\rm tr} \bm A = \sum_{k=1}^K a_{kk}$, and Frobenius norm is $\|\bm H\| = \sqrt{\sum_{r=1,t=1}^{R, T}|h_{rt}|^2}$.


\section{Methods}
\subsection{Channel and System Model}\label{sec_model}

\begin{figure}
    \centering
    \includegraphics[width=0.65\linewidth]{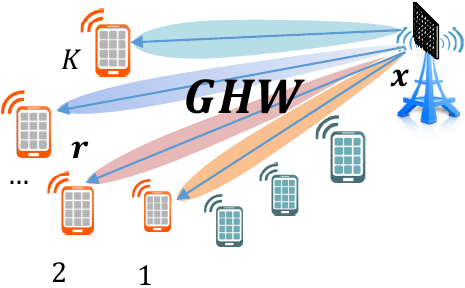}
    \caption{An example of the MIMO precoding usage. The problem is to find an optimal precoding matrix~$\bm W$ of the system given the target Spectral Efficiency (SE) function with the given constraints~\eqref{prob_phys}.}
    \label{fig:mimo_beams}
\end{figure}

\begin{table}
\caption{Symbols and notations\label{tab_notation}}
\label{table_example}
\centering
\begin{tabular}{c c}
\hline
\bf Symbols & \bf Notations\\
\hline\hline
${\bm  ()}^ \mathrm H$ & Complex conjugate operator\\
$\bm H, \bm W$ & Matrices\\
$\bm w_{n}$ & $n$-th column of matrix $\bm W$\\
$\bm h_{m}, \bm w^{m}$ & $m$-th row of matrices $\bm H, \bm W$\\
$h_{nm}, w_{nm}$ & $n,m$-th element of matrices $\bm H, \bm W$\\
$\bm S = \diag(s_1,...,s_N)$ & diagonal matrix  \\
$K$ & the number of users\\
$T$ & the number of transmit antennas\\
$R$ & the total number of receive antennas\\
$R_k$ & the number of receive antennas for each user \\
$L$ & the total number of layers in the system\\
$L_k$ & the number of layers for each user \\

\hline
\end{tabular}
\end{table}

According to~\cite{5G,Tse_tb_05,Bjornson_tb_17} we consider a MIMO broadcast channel. The Multi-User MIMO model is described using the following linear system:
\begin{equation}\label{model}
\bm r = \bm G ( \bm H \bm W \bm x + \bm n ) = \bm G  \bm H \bm W \bm x + \bm G \bm n,
\end{equation}
where $\bm x, \bm r \in \mathbb{C}^L $ are correspondingly \textit{transmitted and received vectors}, $\bm H \in \mathbb{C}^{R \times T}$ is a \textit{downlink channel matrix}, $\bm W \in \mathbb{C}^{T \times L}$ is a \textit{precoding matrix}, and $\bm G \in \mathbb{C}^{L \times R}$ is a block-diagonal \textit{detection matrix}; \textit{noise-vector} $\bm n \sim \mathcal{CN}(0, \sigma^2 \bm I_R) $ is assumed to be independent Gaussian with noise level $\sigma^2$ (Fig.~\ref{fig:mimo_beams}). Note that the linear precoding and detection are implemented by simple matrix multiplications. The constant $T$ is the number of transmit antennas, $R$ is the total number of receive antennas, and $L$ is the total number of transmitted symbols in the system. They are usually related as $L \leqslant R \leqslant T $. Each of the matrices $\bm G, \bm H, \bm W$ decomposes by $K$ \textit{users}: $\bm G = \bdiag\{\bm G_1, \dots, \bm G_K\}, \; \bm H = [\bm H_1; \dots; \bm H_K], \; \bm W = (\bm W_1, \dots, \bm W_K)$ as it is shown on fig.~\ref{Mimo Scheme}, here $\bm G_k\in\mathbb C^{L_k\times R_k}, \bm H_k\in\mathbb C^{R_k\times T}, \bm W_k \in\mathbb C^{T\times L_k}$ are matrices correspond to UE $k$.
    
    

\begin{figure}
  \centering
  \includegraphics[width=1\linewidth]{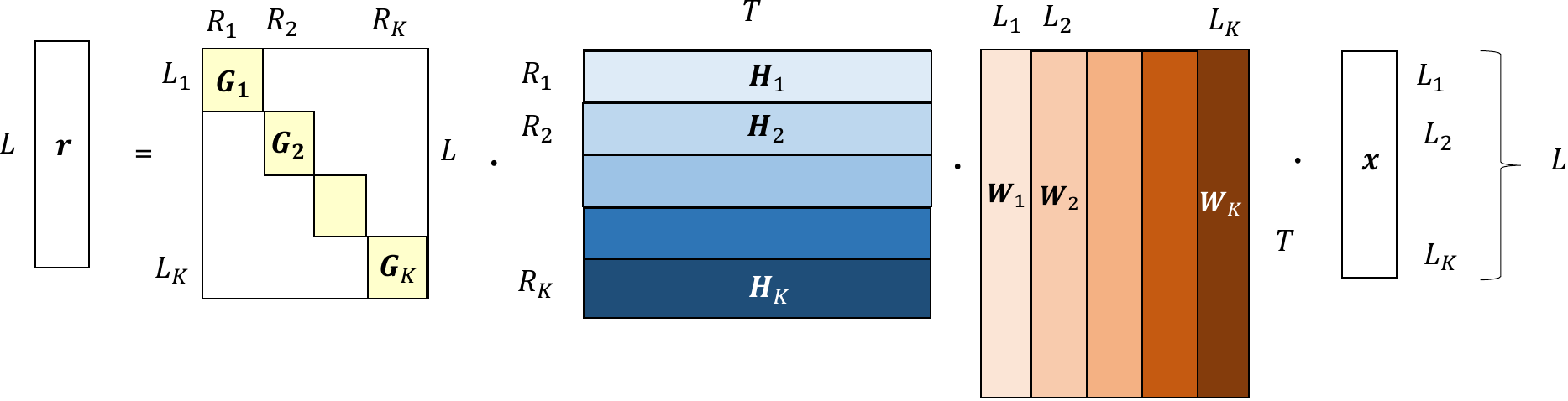}
  \caption{An example of the MIMO transmission system in the matrix form.  Multi-User precoding~$\bm W$ allows transmitting different information to various users simultaneously.}
  \label{Mimo Scheme}
\end{figure}

\subsubsection{Singular Value Decomposition of the Channel}\label{sec_SVD}
It is convenient~\cite{SVD} to represent the channel matrix of UE $k$ via its reduced Singular Value Decomposition (SVD):
\begin{equation}
    \label{User k SVD}
    \bm H_k = \bm U^ \mathrm H_k \bm S_k \bm V_k, \quad \bm U_k \bm U_k^ \mathrm H = \bm U_k^ \mathrm H \bm U_k = \bm I_{R_k}, \quad \bm S_k = \diag\{s_1,\dots,s_{R_k}\}, \quad \bm V_k \bm V^ \mathrm H_k = \bm I_{R_k}.
\end{equation}
Where the \textit{channel matrix} for user $k$, $\bm H_k \in \mathbb{C}^{R_k \times T}$ contains channel vectors $\bm h_i \in \mathbb{C}^{T}$ by rows, the singular values $\bm S_k \in \mathbb{C}^{R_k \times R_k}$ are sorted by descending, $\bm U_k \in \mathbb{C}^{R_k \times R_k}$ is a unitary matrix of left singular vectors, and matrix $\bm V_k \in \mathbb{C}^{R_k \times T}$ consists of \textit{right singular vectors} vector-rows

Collecting all users together, we may write the following \textit{channel matrix} decomposition: $\bm H = \bm U^ \mathrm H \bm S \bm V$ (Lemma~\ref{lem_decomposition} and Fig.~\ref{fig:svd_lemma}), where each of the decomposition matrices $\bm U^ \mathrm H \in \mathbb{C}^{R \times R}, \bm S \in \mathbb{C}^{R \times R}, \bm V \in \mathbb{C}^{R \times T}$ consists of $K$ submatrices of $\bm U_k^ \mathrm H \in \mathbb C ^ {R_k \times R_k}, \bm S_k \in \mathbb{C}^{R_k \times R_k}, \bm V_k \in \mathbb C ^ {R_k \times T}$ which consist of vectors $\bm u_l^ \mathrm H \in \mathbb C ^ {R_k}, s_l \in \mathbb R, \bm v_l \in \mathbb C ^ T$.

\begin{Lemma}[Main Decomposition]\label{lem_decomposition}

For the linear system $\bm r = \bm G ( \bm H \bm W \bm x + \bm n )$ there exists a matrix expansion $\bm H = \bm U^ \mathrm H \bm S \bm V$, 
    Where 
    $\bm S = \text{diag}\{\bm S_1, \dots, \bm S_K\} \in \mathbb{R}^{R \times R}_+$~--- diagonal matrix of singular numbers, $\bm V = [\bm V_1,\dots,\bm V_K] \in \mathbb{C}^{R \times T}$~--- matrix of general form,
    and the matrix $\bm U^ \mathrm H = \text{bdiag} \{\bm U_1^ \mathrm H, \dots, \bm U_K^ \mathrm H\} \in \mathbb{C}^{R \times R}$~--- \textit{block-diagonal and unitary}.

\end{Lemma}

\begin{figure}
    \centering
    \includegraphics[width=\linewidth]{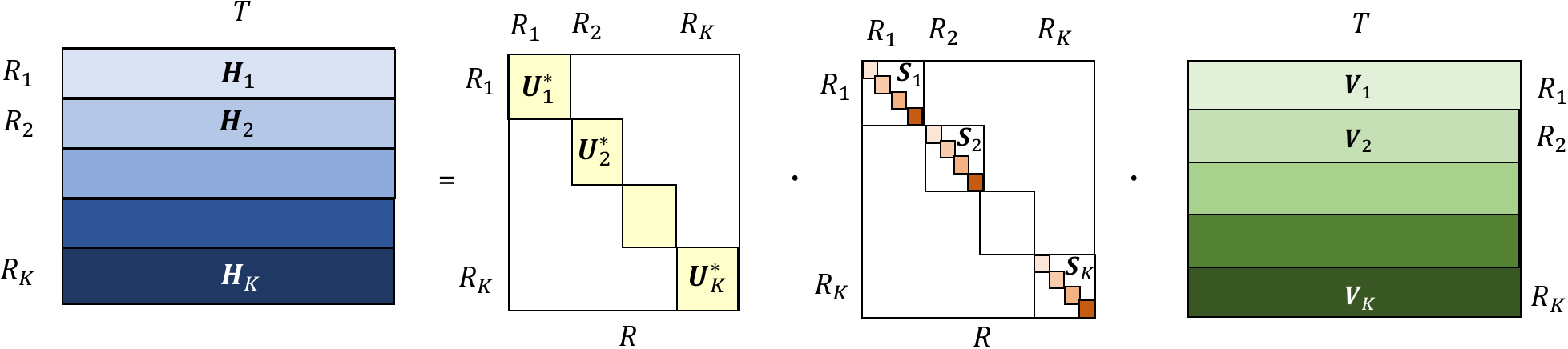}
    \caption{The main decomposition of the channel matrix.  }
    \label{fig:svd_lemma}
\end{figure}


Note that this representation is not actually SVD-decomposition of the matrix $\bm H$: vectors $\bm v_{k,j}, \bm v_{l,i}$ that correspond to different UE $k\neq l$ are not generally orthogonal. Nonetheless, this representation has important properties that make it useful: the matrix $\bm S = \diag(\bm S_k)\in \mathbb{C}^{R \times R}$ is a \textit{diagonal} matrix and $\bm U = \bdiag (\bm U_k) \in \mathbb{C}^{R \times R}$ is a \textit{block-diagonal unitary matrix}. This allows to compensate factor $\bm U^ \mathrm H \bm S$ by detection on UE side (each UE deals with its own $\bm U_k^ \mathrm H \bm S_k$). Thus, on the transmitter side, it is sufficient to invert only the matrix $\bm V$, which in itself is much simpler than the channel $\bm H$: first, its rows have a unit norm and, second, it is a natural object for Rank adaptation problem~\cite{RankSelection}. We heavily use the Lemma~\ref{lem_decomposition} in the following.

\subsubsection{Idea of Proposed Precoding Algorithm}

Let us briefly formulate the resulting algorithm. The intuitive idea of the proposed precoding method is that it should provide maximum signal power and eliminate the so-called inter-user interference. It is known that maximal signal power is provided by Maximum Ratio Transmission (MRT) precoding, $\bm W_{MRT}(\bm V) = \bm V^ \mathrm H$, and the interference is vanished by the Zero-Forcing (ZF) precoding, $\bm W_{ZF}(\bm V) = \bm V^ \mathrm H (\bm V \bm V^ \mathrm H)^{-1}$~\cite{ZF_MRT}. Both of these algorithms have disadvantages and can be improved by using Regularized Zero-Forcing (RZF) precoding $\bm W_{RZF}(\bm V) = \bm V^ \mathrm H (\bm V \bm V^ \mathrm H + \lambda \bm I)^{-1}$, where the regularization parameter $\lambda > 0$ depends on noise level and average path-losses~\cite{Joham_RZF}.

When users have significantly different path-losses, it is better to use regularization with a diagonal matrix. Non-scalar (matrix) regularization is introduced by E.~Bj\"ornson in~\cite{Bjornson}, where it is constructed using the primal-dual technique and demonstrated that the optimal regularization takes the shape of a diagonal matrix (generally speaking, with different elements). This proof is not constructive since it does not propose a specific formula or procedure for optimum regularization. In our work, we propose an explicit heuristic formula for diagonal regularization, Adaptive Regularized Zero-Forcing (ARZF) for the MU-MIMO system, as follows: $$\bm W_{ARZF} = \bm V^ \mathrm H (\bm V \bm V^ \mathrm H + \lambda \bm S^{-2})^{-1}.$$
Theoretical justification is given below (Sec.~\ref{sec_precoding}). Experiments with Quadriga~\cite{Quadriga} show that proposed ARZF outperforms known scalar RZF algorithms (see Sec.~\ref{sec_simulation}).

\subsubsection{Number of Transmitted Streams and Conjugate Detection}\label{sec_ConjDet}

In the previous section, we have introduced the notation of SVD (eq.~\ref{User k SVD}). Usually, the transmitter sends to UE several layers and the number of layers (rank) is less than the number of UE antennas ($L_k \leqslant R_k$). In this case, it is natural to choose for transmission the first $L_k$ vectors from ${\bm V}_k$ that correspond to the $L_k$ largest singular values from ${\bm S}_k$. 
Denote by $\widetilde{\bm S}_k \in \mathbb{C}^{L_k \times L_k}$ the first $L_k$ largest singular values from $\bm S_k$, and by  $\widetilde{\bm U}_k^ \mathrm H \in \mathbb C^{R_k \times L_k}, \; \widetilde {\bm V}_k \in \mathbb C^{L_k \times T}$ the first $L_k$ left and right singular vectors that correspond to $\widetilde{\bm S}_k$:
\begin{equation}
\widetilde{\bm S}_k = \diag\{s_{k,1}, \dots, s_{k,L_k}\}, \quad
\widetilde{\bm U}_k^ \mathrm H = (\bm u_{k,1}^ \mathrm H, \dots, \bm u_{k,L_k}^ \mathrm H), \quad
\widetilde{\bm V}_k = [\bm v_{k,1}; \dots; \bm v_{k,L_k}],
\end{equation}
where the wave sign $\widetilde{\cdot}$ denotes the contraction of singular vectors, i.e. ${\rm rank} \widetilde {\bm V}_k = L_k \leqslant R_k = {\rm rank} \bm V_k$.
Numbers $L_k$ (and particular selection of $\widetilde {\bm V}_k$) are defined during the Rank Adaptation problem that, along with Scheduler, is solved before precoding. 
For the Rank adaptation problem, we refer for example to~\cite{RankSelection} and in what follows we consider $L_k$, $\widetilde {\bm V}_k$ already chosen.

After precoding and transmission, on the UE $k$ side, we have to choose the detection matrix $\bm G_k \in \mathbb C^{L_k \times R_k}$ that takes into account UE Rank $L_k$. The way, how UE performs detection, heavily affects the total performance and different detection algorithms require different optimal precodings (see~\cite{Joham_RZF} where precoding is a function of the detection matrix). The best way would be to choose precoding and detection consistently but it is hardly possible due to the distributed nature of wireless communication. Nevertheless, there are ideas on how to adjust the precoding matrix assuming a particular way of detection on the UE side at the transmitter~\cite{PrecodingDetection}. We do not consider such an approach in this paper, though it can be used to further improve our main proposal.

To conduct analytical calculations, we assume the Conjugate Detection (CD)~\cite{bobrov2023power} in the following form:
\begin{equation}\label{Conjugate Detection}
       \bm G^C_k = \widetilde{\bm S}_k^{-1}  \widetilde{\bm U}_k \in \mathbb C ^ {L_k \times R_k} \Longleftrightarrow \bm G^C := \widetilde{\bm S}^{-1} \widetilde{\bm U} \in \mathbb{C}^{L \times R} ,
\end{equation}
where the unitary matrix $\widetilde{\bm U}_k \in \mathbb C ^{L_k \times R_k} $ contains the first $L_k$ singular vectors. The diagonal matrix $\widetilde{\bm S}_k \in \mathbb{C}^{L_k \times L_k}$ contains the first $L_k$ of the largest singular values. A block-diagonal unitary matrix $\widetilde{\bm U}$ consists of $\widetilde{\bm U}_k$ blocks. The diagonal matrix $\widetilde{\bm S}^{-1}$ consists of blocks $\widetilde{\bm S}_k^{-1}$. Finally, the block-diagonal detection matrix $\bm G^C$ consists of the blocks $\bm G^C_k$ on the main diagonal or, which is the same, the product of the matrices $ \bm G^C := \widetilde{\bm S}^{-1} \widetilde{\bm U} \in \mathbb{C}^{L \times R}$.

\begin{Theorem}\label{theorem_conj_det}
Conjugate Detection deletes unused singular vectors: $\bm G^C \bm H = \widetilde {\bm V}$, and the model equation~\eqref{model} takes the form
\begin{equation}\label{model_conj}
\bm r = \widetilde{\bm V} \bm W \bm x + \tilde{\bm n}, \quad 
\tilde {\bm n} := \widetilde{\bm S}^{-1} \widetilde{\bm U} \bm n.
\end{equation}
\end{Theorem}
\begin{proof}
Using Lemma~(\ref{lem_decomposition}) we can write 
\begin{align}&
    \bm G^C_k \bm H_k = \widetilde{\bm S}^{-1}_k \widetilde{\bm U}_k \bm U^ \mathrm H _ k \bm S _ k \bm V _ k = \widetilde{\bm S}^{-1} _ k \big[\begin{array}{c|c} \bm I _ {L_k} & \bm O \end{array}\big] \bm S _k \bm V _k  = \widetilde{\bm S}^{-1} _k \widetilde{\bm S} _k \widetilde{\bm V} _k = \widetilde{\bm V} _k,
\end{align}
which immediately leads to with combination of users $k=1 \dots K$~\eqref{model_conj}.
\end{proof}
\begin{Corollary}\label{cor_noise}
Under the assumption of independent Gaussian noise $\bm n \sim \mathcal {CN}(0,\sigma^2 \bm I_R)$ the distribution of the \textit{effective noise} vector $\tilde {\bm n} \in \mathbb C^{L\times 1}$ (that appears with Conjugate Detection) is $\tilde {\bm n} \sim \mathcal {CN}(0,\sigma^2 \tilde {\bm S}^{-2})$.
\end{Corollary}
\begin{proof}
\begin{multline}
{\mathbb E} [\tilde {\bm n} \tilde {\bm n}^ \mathrm H] = 
{\mathbb E} [\tilde{\bm S}^{-1} \tilde {\bm U} {\bm n} {\bm n}^ \mathrm H \tilde {\bm U}^ \mathrm H \tilde {\bm S}^{-1}] = \\ =
\tilde{\bm S}^{-1} \tilde {\bm U}  {\mathbb E} [ {\bm n} {\bm n}^ \mathrm H ] \tilde {\bm U}^ \mathrm H \tilde {\bm S}^{-1} =
\tilde{\bm S}^{-1} \tilde {\bm U}  \sigma^2 \bm I_R \tilde {\bm U}^ \mathrm H \tilde {\bm S}^{-1} = 
\sigma^2 \tilde{\bm S}^{-1} \bm I_L \tilde{\bm S}^{-1} =
\sigma^2 \tilde {\bm S}^{-2}. 
\end{multline}
\end{proof}

\begin{Remark}\label{rem_noise} 
The Corollary~\ref{cor_noise} is essential for our study, since it gives the idea to use $\widetilde{\bm S}^{-2} \in \mathbb C ^ {L \times L}$ in regularization part of precoding to take into account the correct effective noise $\widetilde{\bm S}^{-1} \widetilde{\bm U} \bm n$.
\end{Remark}
\begin{Remark}\label{Reducing to Layers}
The formulated Theorem sufficiently simplifies the initial problem, decreases its dimensions, and allows notation to be uniform. Namely, we can work with \textit{user layers} of shapes $L_k$ and $L$ instead of considering \textit{user antennas} space.
Note also that it is sufficient to only perform Partial SVD of the channel $\bm H_k \in \mathbb C ^ {R_k \times T}$, keeping just the first $L_k$ singular values and vectors for each user $k$ such as 
\begin{equation}
\bm H_k \approx \widetilde{\bm U}^ \mathrm H_k \widetilde{\bm S}_k \widetilde{\bm V}_k.
\end{equation}
Based on this, in what follows we omit the tilde and write $\bm U_k, \bm S_k, \bm V_k$ instead of $\widetilde{\bm U}_k, \widetilde{\bm S}_k, \widetilde{\bm V}_k$ correspondingly.
\end{Remark} 
\begin{Remark}
The introduced Conjugate Detection is ``ideal'' and can not be implemented in practice. However, one can show that realistic detection policies as MMSE or IRC detection~\cite{IRC} often behave similarly to Conjugate Detection. In simulations, we use MMSE detection to compare various precoding methods. 
\end{Remark}


\subsubsection{Known MMSE Detection}

Signal detection aims at pseudo-reversing the product of the channel and precoding matrices. The most common form of the $\bm G$ detection matrix is studied, which are called minimum MSE (MMSE)~\cite{MMSE}. In statistics and signal processing, MMSE~--- is an estimation method that minimizes the MSE function, which in turn is a general measure of the quality of the estimation of fitted values of the dependent variable. 

The definition of MMSE detection implements the following rule:
\begin{equation}\label{MMSE_1}
    \bm G^{\textit{MMSE}}_k(\bm A_k) = \bm A_k ^ \mathrm H \left(\bm A_k \bm A_k^ \mathrm H + \sigma^2 \bm I\right)^{-1}, \quad \bm A_k = \bm H_k \bm W_k.
\end{equation}
The parameter $P$~--- is the base station power, and $\sigma^2$~--- the system noise. The MMSE detection method seeks to eliminate noise by assuming it is the same for all symbols: $\bm n \sim \mathcal{CN}(0,\sigma^2 \bm I_R)$, which may be violated in practice.

\begin{Lemma}
    For system~\eqref{model}: $\bm r_k = \bm G_k \bm H_k \bm W \bm x + \bm G_k \bm n_k \in \mathbb{C}^{L_k}$ with noise distribution $\bm n \sim \mathcal{CN}(0, \sigma^2 \bm I_R)$, and precoding $\bm W = \bm H^+ = \bm H^* (\bm H \bm H^*)^{-1}$, matrix~\eqref{MMSE_1}: $\bm G^{\textit{MMSE}}$ minimizes the square of the norm: $ \mathbb E_{\bm n, \bm x} \|\bm r_k - \bm x \|^2$, $ k=1 \dots K$.
\end{Lemma}
\begin{proof}
    
    Let's substitute the system expression $\bm r_k = \bm G_k \bm H_k \bm W \bm x + \bm G_k \bm n_k$ into the loss function $\mathbb E_{\bm n \sim \mathcal{CN}(0,\sigma^2 \bm I_R)} \|\bm r_k - \bm x \|^2 = \mathbb E_{\bm n \sim \mathcal{CN}(0,\sigma^2 \bm I_R)}\|\bm G_k \bm H_k \bm W \bm x - \bm x + \bm G_k \bm n_k \|^2$.

The second summand will go to zero from introducing the expectation of noise with zero mean, and the third summand will go to zero 

     Let's open the brackets using the sum norm squared formula: 
    \begin{multline}
        \mathbb E_{\bm n \sim \mathcal{CN}(0,\sigma^2 \bm I_R)} \|(\bm G_k \bm H_k \bm W - \bm I) \bm x + \bm G_k \bm n_k\| ^ 2  = \|(\bm G_k \bm H_k \bm W - \bm I) \bm x\|^2 + \\ + 2 \mathbb E_{\bm n \sim \mathcal{CN}(0,\sigma^2 \bm I_R)} \Re \{   < \bm G_k \bm H_k \bm W \bm x - \bm x, \bm G_k \bm n_k > \} +  \mathbb E_{\bm n \sim \mathcal{CN}(0,\sigma^2 \bm I_R)} \| \bm G_k \bm n_k \|^2 = \\ = \|(\bm G_k \bm H_k \bm W - \bm I) \bm x\|^2 + 2  \Re \{   < \bm G_k \bm H_k \bm W \bm x - \bm x, \bm G_k \underbrace{\mathbb E_{\bm n \sim \mathcal{CN}(0,\sigma^2 \bm I_R)} \bm n_k >}_{=0} \} + \sigma^2 \| \bm G_k \|^2 = \\ = \|(\bm G_k \bm H_k \bm W - \bm I) \bm x\|^2 + \sigma^2 \| \bm G_k \|^2
    \end{multline}
    Next, applying expectation over the symbols $\bm x$ with the condition $\mathbb E _ {\bm x \sim \mathcal{CN}(0, \bm I_L)} \bm x \bm x^* = \bm I$: 
    \begin{multline}
        \mathbb E _ {\bm x \sim \mathcal{CN}(0, \bm I_L)} \|(\bm G_k \bm H_k \bm W - \bm I) \bm x\|^2 + \sigma^2 \| \bm G_k \|^2 = \|\bm G_k \bm H_k \bm W - \bm I \| ^ 2 + \sigma^2 \| \bm G_k \|^2 \rightarrow \min\limits_{\bm G_k}
    \end{multline}

    If the conditions are met: $ \bm H_k \bm W = \{ \bm W = \bm H^* (\bm H \bm H^*)^{-1} \} = \bm H_k \bm W_k$:
    \begin{equation}\label{eq:upper_bound_mse_2}
        \|\bm G_k \bm H_k \bm W_k - \bm I \| ^ 2 + \sigma^2 \| \bm G_k \|^2 \rightarrow \min\limits_{\bm G_k}
    \end{equation}

    Let's calculate the gradient of the function~\eqref{eq:upper_bound_mse_2} and equate it to zero:
    \begin{multline}
            \nabla_{\bm G_k} \{ \|\bm G_k \bm H_k \bm W_k - \bm I \| ^ 2 + \sigma^2 \| \bm G_k \|^2 \} =  2(\bm G_k \bm H_k \bm W_k - \bm I)(\bm H_k \bm W_k)^ \mathrm H + 2\sigma^2 \bm G_k = \\ =
            2\bm G_k (\bm H_k \bm W_k) ( \bm H_k \bm W_k )^ \mathrm H - 2( \bm H_k \bm W_k )^ \mathrm H + 2\sigma^2 \bm G_k = 0 \\
            \bm G_k ((\bm H_k \bm W_k) ( \bm H_k \bm W_k )^ \mathrm H + \sigma^2 \bm I ) = ( \bm H_k \bm W_k )^ \mathrm H \\
            \widehat{\bm G}^{\textit{MMSE}}_k = \bm G_k = ( \bm H_k \bm W_k )^ \mathrm H ((\bm H_k \bm W_k) ( \bm H_k \bm W_k )^ \mathrm H + \sigma^2 \bm I )^{-1}
    \end{multline}

    We get the desired solution~\eqref{MMSE_1}.

\end{proof}

\subsubsection{Quality Measures}\label{sec_QM}


The following functions are used to measure the quality of the precoding methods. These functions are based not on the actual sending symbols $\bm x \in \mathbb C ^ {L\times 1}$, but some distribution of them~\cite{Bjornson}. Thus, we get the common function for all assumed symbols, which can be sent using the specified precoding matrix.

Let us consider the end-to-end numbering of symbols $l= 1,...,L$ for all UE and index function $k = k(l)$ that returns the index of UE that receives the symbol $l$.  
The \textit{Signal-to-Interference-and-Noise} functional of the $l$-th symbol of user $k = k(l)$ is defined as: 
\begin{equation}\label{Symbol SINR}
\textit{SINR}_l(\bm W, \bm H_k, \bm g_l, \sigma^2) := \dfrac{|\bm g_l \bm H_k \bm w_l |^2}{\sum_{i \ne l}^{L} | \bm g_l \bm H_k \bm w_i |^2 + \sigma^2 \|\bm g_l\|^2}.
\end{equation}

The formula (\ref{Symbol SINR}) shows the ratio between the useful and harmful parts of the signal. It depends on the whole precoding matrix $\bm W \in \mathbb C ^ {T \times L}$, where the complex vector $\bm w_l \in \mathbb C ^ {T\times 1}$ denotes the precoding for the $l$-th symbol, on the channel matrix $\bm H_k \in \mathbb C ^ {R_k \times T}$ of the $k$-th user, the detection vector $\bm g_l \in \mathbb C ^ {1\times R_k}$ of the $l$-th symbol; after detection noise $\bm n_k \in\mathbb C^{R_k\times 1}$ level of UE $k$ becomes $\mathbb E[\bm g_l \bm n] = \sigma^2 \|\bm g_l\|^2$. 
The formula (\ref{Symbol SINR}) can be efficiently computed for all $L$ layers using several matrix multiplications and summations.
It can be further simplified, using Theorem~\ref{theorem_conj_det}.
\begin{Corollary}
For given SVD~\eqref{User k SVD} and assuming Conjugate Detection~\eqref{Conjugate Detection} the following formula for $\textit{SINR}$ holds:
    \begin{equation}\label{SINR_conj}
    \textit{SINR}_l^{C}(\bm W, \bm v_l, s_l, \sigma^2) = \textit{SINR}_l(\bm W, \bm H_k, \bm g_l^C, \sigma^2) =  \frac{ | \bm v_l \bm w_l  |^2 }{\sum_{i \ne l}^L| \bm v_l \bm w_i |^2  +  \sigma^2 / s_l^{2}}.
    \end{equation}
\end{Corollary}


The important criterion of network performance is the \emph{spectral efficiency} $\textit{SE}_k$ of UE $k$, which refers to the information rate that can be transmitted over a given bandwidth for a certain UE. In the case of one symbol $L_k=1$ it is bounded by Shannon's entropy that depends on $\textit{SINR}$ in the following way:
\begin{align}\label{f:sek}
\mathcal S (\textit{SINR}) :=  \log_2 (1 + \textit{SINR}).
\end{align}
This is theoretically supreme for the possible transmission Rate,  
and recent modulation and coding schemes along with Hybrid Automatic Repeat Request (HARQ) retransmission and Block Error Rate (BLER) management allow us to achieve a Rate very close to Shannon's entropy. Note that $\textit{SINR}$ here is taken in linear values, not in dB.

In the case of several transmitted symbols, generally speaking, $\textit{SE}_k$ of UE $k$ is not a sum over its layers because there is one common transport block to be transmitted and thus the coding and modulation algorithms are also common. Usually, one would introduce an \emph{effective} \textit{SINR} as a function of per-symbol \textit{SINR}:
$\textit{SINR}_k^{eff} = f(\textit{SINR}_1, \dots, \textit{SINR}_{L_k})$, and so using~\eqref{f:sek} we obtain:
\begin{align}& \label{f:SE}
\textit{SE}_k(\bm W,\bm H_k, \bm G_k,\sigma^2) = L_k \; \mathcal S \Big(\textit{SINR}_k^{eff}(\bm W,\bm H_k, \bm G_k,\sigma^2)\Big).
\end{align}
There are different approaches to estimate this effective SINR for QAM64 and QAM256 (see~\cite{effectiveSINR}) and in simulations, we use the QAM256 model. Also, we consider approximate formula using the geometric mean of per-symbol SINR. That is the same as the usual average of per-symbols in dB. This heuristic approximation will be used later in numerical optimization by gradient search:
\begin{equation}\label{f:SINReff}
\textit{SINR}_k^{eff}(\bm W, \bm H_k, \bm G_k, \sigma^2) \approx \Big({\prod\nolimits_{l \in \mathcal{L}_k} \textit{SINR}_l(\bm W, \bm H_k, \bm g_l, \sigma^2) } \Big)^{\frac{1}{L_k}}.
\end{equation}


The most general problem statement is the multi-criteria optimization of the whole vector $(\textit{SE}_1,...,\textit{SE}_K)$. For such a problem, the Pareto optimality can be studied (see~\cite{Pareto1,Pareto2}), which is hard and does not provide a unique solution, that is why usually a suitable decomposition into one-criterial optimization is considered: $J = J(\textit{SE}_1,...,\textit{SE}_K) \to \max_{}$ or $J = J(\textit{SINR}_1,...,\textit{SINR}_K)$~\cite{Bjornson}. Such decomposition can be done in different ways, we consider the sum of Spectral Efficiencies~\eqref{f:SE} over all UEs:
\begin{equation}\label{J_SE}
J^{SE}(\bm W) := \textit{SE}(\bm W,\bm H, \bm G,\sigma^2) = \sum_{k=1}^K \textit{SE}_k(\bm W,\bm H_k, \bm G_k,\sigma^2).
\end{equation}
Such criterion is natural as UE Rates are additive values. There are other possible targets, e.g. performance of cell edge UE (CEU). In~\cite[sec.~7]{Bjornson_tb_17} interested readers can find Pareto analysis of the multi-criteria statement and comparison of several target functions, including
\begin{align}\label{Min SE}
\textit{SE}_{\min} = \min_k \textit{SE}_k \rightarrow \max{} \quad \text{or} \quad 
\textit{SINR}_{\min} = \min_{1\leqslant j\leqslant L} \textit{SINR}_j \rightarrow \max{}.
\end{align}


Finally, we consider Single-User \textit{SINR} for the $k$-th user and the average of \textit{SU SINR} is an important parameter in simulations:
\begin{equation}
    \textit{SUSINR}_k({\bm S}_k, \sigma^2, P) :=  \frac{P}{L_k\sigma^2} \bigg(\prod\nolimits_{l \in \mathcal{L}_k} s_l^2 \bigg)^{\frac{1}{L_k}}, \quad
\end{equation}
\begin{equation}
    \textit{AvSUSINR}(\bm S, \sigma^2, P) := \bigg(\prod_{k=1}^K \textit{SUSINR}_k({\bm S}_k, \sigma^2, P)\bigg)^{\frac 1 K}.  
    \label{SU SINR}
\end{equation}
The formula (\ref{SU SINR}) reflects the quality of the channel (geometrical average of per-symbol \textit{SINR}) for the specified UE without taking into account other users. It depends on the greatest $L_k$ singular values ${\bm S}_k \in \mathbb R ^ {L_k \times L_k}$ of the $k$-th user channel matrix $\bm H_k \in \mathbb{C}^{R_k \times T}$ and may be derived from the (\ref{Symbol SINR}) and (\ref{f:SINReff}) formulas assuming Single-User case, MRT or ZF precoding matrix and Conjugate Detection (\ref{Conjugate Detection}). We will use this function in our experiments as a universal channel characteristic, including system noise $\sigma^2$ and station power $P$.

\subsubsection{Problem Statement and Power Constraints} \label{sec_PC}

First of all, we assume that the total channel $\bm H$, the number $K$ of UE, and their ranks $L_k$ are known given values. This means that the Scheduler problem (which UE are to be served from the set of active UE) and Rank adaptation problem (which rank is provided to each UE) are already solved. This is usually the case in real networks. Scheduler and Rank adaptation problems are complicated and important radio resource management problems themselves but are out of the scope of this study (for example of Scheduler problem we refer to~\cite{Scheduler} and bibliography within, for Rank adaptation,~\cite{RankSelection} and~\cite{PrecodingDetection}.
These algorithms affect the properties of matrix $\bm H$, e.g. Scheduler can choose only UE with small enough correlations $\|\bm C \| = \|\bm V \bm V^ \mathrm H - \bm I_R\| \leqslant \varepsilon$. We take this into account and consider scenarios with small correlations of UE channels.

Next, we consider the channel model in the form~\eqref{model} that particularly means exact measurements of the channel. To further simplify the problem we suppose detection policy $\bm G = \bm G(\bm H, \bm W)$ to be a known function, moreover we assume Conjugate Detection~\eqref{Conjugate Detection} that simplifies channel model to~\eqref{model_conj}. Based on this channel model we calculate \textit{SINR} of transmitted symbols by~\eqref{SINR_conj} and effective \textit{SINR} of UE, which can be approximately calculated by~\eqref{f:SINReff}.

Finally, denote the total power of the system as $P$ and that the sent vector has unit norm $\mathbb{E}[\bm x \bm x^ \mathrm H] = \bm I_L$. The \textit{total power constraints}
and the more realistic \textit{per-antenna power constraints} (see~\cite{Bjornson_tb_17}) impose the following conditions on the precoding matrix:
\begin{equation}\label{power_antenna}
\|\bm W\|^2 \leqslant P, \qquad \text{or} \qquad 
\|(w_{t1},\dots,w_{tL})\|^2 \leqslant P/T, \quad t=1,...,T.
\end{equation}

{\bf The ultimate goal is to find a precoding matrix that maximizes sum Spectral Efficiency~\eqref{f:SE}
subject to power constraints~\eqref{power_antenna}}, e.g.: 
\begin{equation}\label{prob_phys}
J^{SE}(\bm W) := SE^C(\bm W), \quad
\bm W = {\rm argmax}_{\bm W} J^{SE}(\bm W), \quad {\rm s.t.:} \;\; \|\bm W\|^2\leqslant P.
\end{equation}

Even after all the above simplifications, the formulated problem is too complicated to solve analytically. Moreover, it is not convex or concave, hence it could have a lot of (essentially) different local maximums. That's why our strategy in this paper is to propose some heuristic formula that will prove to be better than known algorithms on some reliable simulations. 
After defining the particular form of precoding (anzats) $\bm W^0 = \bm W^0(\bm V, \bm S, \bm n)$ we can always satisfy power constraints by normalizing constant, e.g. for~\eqref{power_antenna}:
\begin{equation}\label{mu}
\bm W = \mu \bm W^0, \quad    
\mu = \dfrac{\sqrt{P}}{\|\bm W^0\|} \quad {\rm or} \quad \mu = \dfrac{\sqrt{P/T}}{\max\limits_{t=1,...,T} \{\|(w_{t1},\dots,w_{tL})\| \} }.
\end{equation}
In simulations, we use more realistic per-antenna power constraints.

A heuristic algorithm is formulated below and its idea is hinted at by formulated simplifications in Corollary~\ref{cor_noise} and Remark~\ref{rem_noise}. Theoretically, it is supported by the model problem of MSE minimization (see e.g.~\cite{Joham_RZF}):
\begin{equation}\label{prob_constr}
J^{MSE}(\bm W) := \mathbb E_{\bm x, \bm n} [\|\bm r(\bm W) - \bm x\|^2], \quad
\bm W = {\rm argmin}_{\bm W} J^{MSE}(\bm W), 
\quad {\rm s.t.:}\;\; \|\bm W\|^2 \leqslant P,
\end{equation}
where $\bm r(\bm W)$ is given by channel model~\eqref{model} or~\eqref{model_conj}.

\subsection{Reference Precoding Methods}\label{sec_precoding}

We repeat known reference precoding algorithms and present our solution.

\subsubsection{Maximum Ratio Transmission (MRT)}

A Maximum Ratio Transmission precoding algorithm takes single-user weights $\bm V^ \mathrm H$ from the SVD-decomposition. This approach leads to the maximization of single-user power, ignoring the interference. The MRT approach is preferred in noisy systems when the noise power is higher than inter-user interference~\cite{ZF_MRT}: 
\begin{equation}
\bm W_{MRT}(\bm V) = \mu \bm V^ \mathrm H,
\end{equation}
where normalizing constant $\mu$ is defined from power constraint~\eqref{mu}.
This algorithm outputs a set of interfering channels assuming the model~\eqref{model_conj}:
$$\bm r = \bm V \bm W \bm x + \bm S^{-1} \bm U \bm n = \mu \bm V \bm V^ \mathrm H  \bm x + \bm S ^{-1} \bm U \bm n.$$



\subsubsection{Zero-Forcing (ZF)}

The next modification of the precoding algorithm performs decorrelation of the symbols using the inverse correlation matrix of the channel vectors. Such precoding construction sends the signal beams to the users without creating any interference between them. Different from the MRT method, the Zero-Forcing approach is preferred when the potential inter-user interference is higher than the noise power, and the Spectral Efficiency quality improves by eliminating this interference~\cite{ZF_MRT}.
\begin{equation}\label{Zero-Forcing}
    \bm W_{ZF}(\bm V) = \mu \bm V^\dagger = \mu \bm V^ \mathrm H (\bm V \bm V^ \mathrm H)^{-1}.
\end{equation}

The following receiver model is:     $$\bm r =  \bm V \bm W \bm x + \bm S^{-1} \bm U \bm n =  \mu  \bm V \bm V^ \mathrm H (\bm V \bm V^ \mathrm H)^{-1}  \bm x + \bm S^{-1} \bm U \bm n = \mu \bm x + \bm S^{-1} \bm U \bm n$$






Denote $\bm F = \bm U \bm H = \bm S \bm V$.
\begin{Theorem}\label{theorem_ZF}
Assume the channel matrix has the form of $\bm H = \bm U^ \mathrm H \bm S \bm V$ (Lemma~\ref{lem_decomposition}). Then, the following relation for Zero-Forcing holds:
\begin{equation}\label{ZF-ZF}
\bm W_{ZF}(\bm F) \bm S = \bm W_{ZF}(\bm V)
\end{equation}
\end{Theorem} 

\begin{proof}
\begin{multline}
\bm W_{ZF} (\bm F) \bm S = \bm F^ \mathrm H (\bm F \bm F^ \mathrm H )^{-1} \bm S = \bm V^ \mathrm H \bm S( \bm S \bm V \bm V^ \mathrm H \bm S)^{-1} \bm S = \\  = \bm V^ \mathrm H \bm S \bm S^{-1} (\bm V \bm V^ \mathrm H)^{-1} \bm S^{-1} \bm S = \bm V^ \mathrm H (\bm V \bm V^ \mathrm H )^{-1}  = \bm W_{ZF} (\bm V) \\
\end{multline}
\end{proof}


\subsubsection{Regularized Zero-Forcing (RZF)}

In the geometrical sense, in the ZF method (\ref{Zero-Forcing}), beams are sent not directly to the users but with some deviation, which reduces the useful signal. The following modification corrects the beams, which allows some inter-user interference and significantly increases the payload. 

In the practical sense, in the ZF method (\ref{Zero-Forcing}), the channel right inversion may not exist or matrix $\bm V \bm V^ \mathrm H$ may be ill-conditioned, making ZF poorly perform. There are many practical solutions to this problem based on regularization.
\begin{equation}\label{RZF}
\bm W_{RZF} (\bm V)  = \mu \bm V^ \mathrm H(\bm V \bm V^ \mathrm H +  \lambda \bm I )^{-1}
\end{equation}

Regularized Zero-Forcing is the most common method in real practice, and therefore we use it as the main reference method. As the baseline, we use the analytical form of the regularization matrix using $\lambda = \frac{L\sigma^2}{P}$~\cite{OptimalRegularization}. 

This method cannot cancel all multi-user and multi-layer interference. It admits some interference to maximize single-user power. It is used as a trade-off between using MRT and ZF precoding~\cite{Bjornson} balancing between maximizing the signal power and minimizing the interference leakage, and thus we need to appropriately manage them by optimizing the regularization parameter depending on the noise level. 

The RZF method has the following asymptotic properties~\cite{ZF_MRT}: if $\sigma^2 \rightarrow \infty$, it becomes equivalent to $\bm W_{MRT} = \mu \bm V^ \mathrm H$, which is optimal in \textit{low} \textit{SINR} cases. And if we set $\sigma^2 = 0$, the formula becomes equal to ZF precoding: $\bm W_{ZF}  = \mu \bm V^ \mathrm H(\bm V \bm V^ \mathrm H)^{-1} $, which is optimal in \textit{high} \textit{SINR} cases.

The precoding matrix based on the un-normalized channel~\cite{RZF19} in the case when the number of sending symbols $L$ is less than the number of receiver antennas may be written in the following form of RZF:
\begin{equation}
\bm W_{RZF}(\bm F)  =  \mu \bm F^ \mathrm H(\bm F \bm F^ \mathrm H + \lambda \bm I)^{-1}.
\end{equation}
Here, $\bm F = \bm S \bm V$ parameter and $\lambda = \frac{L\sigma^2}{P}$~\cite{OptimalRegularization} are chosen taking into account noise level $E[\bm n \bm n^ \mathrm H] = \sigma^2 \bm I_R$. Actually, this method is equivalent to an un-normalized channel matrix $\bm H$, which in our case is the matrix $\bm F$, obtained after splitting user channel into multiple streams.

Let us formulate the following well-known fact about 
RZF method~\eqref{RZF}. 
\begin{Theorem}\label{theorem_RZF} Consider the channel decomposition $\bm H = \bm U^ \mathrm H \bm S \bm V$ from the Lemma~\ref{lem_decomposition}. The precoding $\bm W_{RZF}(\bm V)$ with any parameter $\lambda>0$ 
is the solution of the following optimization problem:
\begin{align}& \label{J}
\bm W_{RZF}(\bm V) = {\rm argmin}_{\bm W} J(\bm W), \quad {\rm where} \;\;  J(\bm W) = \| \bm V \bm W - \bm I \|^2_2 + \lambda \| \bm W \|^2_2.
\end{align}
\end{Theorem}

\begin{proof}
Calculating the gradient and equating it to zero, we get:
\begin{align} 
\nabla J(\bm W) = 2 \bm V^ \mathrm H (\bm V\bm W - \bm I) + 2 \lambda \bm W = 0 \quad \Leftrightarrow \quad &
(\bm V^ \mathrm H \bm V + \lambda \bm I) \bm W = \bm V^ \mathrm H
\\  \nonumber
\quad  \Leftrightarrow \quad  
\bm W = (\bm V^ \mathrm H \bm V + \lambda \bm I)^{-1} \bm V^ \mathrm H = \bm V^ \mathrm H (\bm V\bm V^ \mathrm H + \lambda \bm I)^{-1}.
\end{align}

The last identity may be proved using multiplication by the $(\bm V\bm V^ \mathrm H + \lambda \bm I)$ matrix from the right side of the identity and by the $(\bm V^ \mathrm H \bm V + \lambda \bm I)$ from the left side of it.


\end{proof}

\begin{Remark}
The algorithm $\bm W_{RZF}(\bm V)$ is also the solution to the constrained optimization problem~\eqref{prob_constr} in assumption $\bm G \bm H = \bm V$ (see~\cite{Joham_RZF}):
\begin{align}\label{prob_constr_opt}
\bm W_{RZF}(\bm V) = {\rm argmin}_{\bm W} {\mathbb E}_{\bm x, \bm n} \big[ \|\bm V \bm W \bm x - \bm x + \bm n\|^2 \big], \quad {\rm s.t.:} \;\; \|\bm W\|^2\leqslant P.
\end{align}
Constraint optimization problem~\eqref{prob_constr_opt} is decomposed to~\eqref{J} with Lagrangian multiplier $\lambda = L\sigma^2 / P$. In other words, $\bm W_{RZF}(\bm V)$ is a special case of the Wiener filter~\cite[Eq.~(34),(35)]{Joham_RZF}, when covariance matrices of signal and noise are $\bm R_{\bm x} := {\mathbb E}[ \bm x \bm x^ \mathrm H ] = \bm I_L$ and $\bm R_{\bm n} := {\mathbb E}[ \bm n \bm n^ \mathrm H ] = \sigma^2 \bm I_R$ and at assumption of special detection $\bm G = \bm U$. 
\end{Remark}

\subsubsection{Wiener Filter Zero-Forcing (WRZF)}

In~\cite{Joham_RZF} authors study the Wiener filter that provides optimum in the problem~\eqref{prob_constr_opt} for arbitrary symbol and noise covariance matrices $\bm R_{\bm x}$ and $\bm R_{\bm n}$ in the case of known detection matrix $\bm G$. We take Conjugate Detection~\eqref{Conjugate Detection} $\bm G = \bm G^C$ and apply the Wiener filter for normalized channel $\bm V$ and symbol $\bm R_{\bm x} = \bm I_{L}$ and corresponding noise covariance $\bm R_{\bm n} = \bm S^{-2}$.
The algorithm 
\begin{equation}\label{WRZF}
\bm W_{WRZF} (\bm V, \bm S)  = \mu \bm V^ \mathrm H(\bm V \bm V^ \mathrm H +  \lambda \bm I )^{-1}, \quad \lambda = \frac {\sigma^2}{P} {\rm tr} (\bm S^{-2})
\end{equation}
is the solution to the constrained optimization problem~\eqref{prob_constr} in assumption of CD ($\bm G = \bm G^C$ so that $\bm G^C \bm H = \bm V$ according to Theorem~\ref{theorem_conj_det}):
\begin{align}&
\bm W_{WRZF}(\bm V, \bm S) = {\rm argmin}_{\bm W} {\mathbb E}_{\bm x, \bm n} \big[ \|\bm V \bm W \bm x - \bm x + \bm n\|^2 \big], \quad  
{\rm s.t.:} \;\; \|\bm W\|^2 \leqslant P,
\end{align}
where $\mathbb E [\bm n\bm n^ \mathrm H ] = \bm S^{-2}, \quad 
\mathbb E [\bm x\bm x^ \mathrm H ] = \bm I_L$.

\subsection{Proposed Precoding Methods}

\subsubsection{Proposed Adaptive Regularized Zero-Forcing (ARZF)}

Algorithms $\bm W_{RZF}$ and $\bm W_{WRZF}$ (and even the Wiener filter in more general cases) are precodings with \emph{scalar Regularization}. Taking into account effective noise from Corollary~\ref{cor_noise}, we propose \textit{Adaptive RZF (ARZF)} algorithm with \emph{diagonal matrix regularization}:
\begin{equation}\label{ARZF}
\bm W_{ARZF}(\bm V) = \mu \bm V^ \mathrm H(\bm V \bm V^ \mathrm H + \lambda \bm S^{-2})^{-1}, \quad \lambda = \frac{\sigma^2 L}{P}.
\end{equation}

ARZF allows the application of different regularization for different users and layers corresponding to their singular values that also include UE path loss.

Using $\bm r = \widetilde{\bm V} \bm W \bm x + \widetilde{\bm S}^{-1} \widetilde{\bm U} \bm n$, let's write out the quadratic error of the received and sent symbols $\bm r$ and $\bm x$:
\begin{multline}
        \|\bm r - \bm x \|^2 =  \| \widetilde{\bm V} \bm W \bm x - \bm x + \widetilde{\bm S}^{-1} \widetilde{\bm U} \bm n \| ^2 \\
        \Rightarrow \mathbb E _ {\bm n \sim \mathcal{CN}(0, \sigma^2 \bm I)} \| (\widetilde{\bm V} \bm W - \bm I ) \bm x + \widetilde{\bm S}^{-1} \widetilde{\bm U} \bm n \| ^2 = \| (\widetilde{\bm V} \bm W - \bm I ) \bm x \| ^2  \\
        \Rightarrow \mathbb E _ {\bm x \sim \mathcal{CN}(0, \bm I)}  \| (\widetilde{\bm V} \bm W - \bm I ) \bm x \| ^2  = \| \widetilde{\bm V} \bm W - \bm I \| ^2 
\end{multline}


Let's introduce the inverse noise covariance matrix~$\widetilde{\bm S}$ into the definition of norm, and we get the following weighted least squares function~\eqref{eq:MSE_S}:
    \begin{equation}\label{eq:MSE_S}
        \mathrm{MSE}_S(\bm W)  = \|\widetilde{\bm V}\bm W - \bm I\|_{\widetilde{\bm S}}^2 + \lambda \|\bm W \|^2_2  = \| \widetilde{\bm S} ( \widetilde{\bm V} \bm W - \bm I ) \|^2_2 + \lambda \| \bm W \|^2_2 
    \end{equation}

So, ARZF provides optimum to the problem~\ref{eq:MSE_S} (Theorem~\ref{theorem_ARZF}).

\begin{Theorem}\label{theorem_ARZF}
Consider the channel decomposition $\bm H = \bm U^ \mathrm H \bm S \bm V$ from Lemma~\ref{lem_decomposition}. Precoding~\eqref{ARZF} is the solution of the following optimization problem (weighted MSE with regularization):
\begin{equation}\label{prob_quadratic_S}
\bm W_{ARZF}(\bm V) = {\rm argmin}_{\bm W} J_{\bm S}(\bm W), \quad {\rm where} \;\;
J_{\bm S}(\bm W) := 
\| \bm S ( \bm V \bm W - \bm I ) \|^2 + \lambda \| \bm W \|^2.
\end{equation}
\end{Theorem}

\begin{proof}
Calculating the gradient and equating it to zero, we get:
\begin{multline}
\nabla J_{\bm S}(\bm W) = 2 \bm V^ \mathrm H \bm S(\bm S\bm V\bm W - \bm S) + 2 \lambda \bm W = 0 \quad \Leftrightarrow \quad (\bm V^ \mathrm H \bm S^2 \bm V + \lambda \bm I) \bm W = \bm V^ \mathrm H \bm S^2,
\\ \bm W = (\bm V^ \mathrm H \bm S^2 \bm V + \lambda \bm I)^{-1} \bm V^ \mathrm H \bm S^2 = \bm V^ \mathrm H (\bm V\bm V^ \mathrm H + \lambda \bm S^{-2})^{-1}.
\end{multline}

The last identity may be proved using multiplication by the $(\bm V\bm V^ \mathrm H + \lambda \bm S^{-2})$ matrix from the right side of the identity and by the $(\bm V^ \mathrm H \bm S^2 \bm V + \lambda \bm I)$ from the left side.
\end{proof}
\begin{Remark}
Optimization problem~\eqref{prob_quadratic_S} is not common, and it can not be formulated in the form similar to~\eqref{prob_constr_opt}. Nevertheless, such a problem statement is a more adequate approximation to sum SE maximization problem~\eqref{prob_phys}: because ARZF provides a larger sum SE than RZF of WRZF. Detailed simulation results are given below.
\end{Remark}

\begin{Remark}
The proposed algorithm $\bm W_{ARZF}(\bm V)$ along with Theorem~\ref{theorem_ARZF} is the main result of the paper. Transmit Wiener Filter~\cite{Joham_RZF} algorithms have scalar regularization of the form $\lambda \bm I$, so ARZF is not from this class. 
As it is shown in~\cite{Bjornson}, the maximum of the function of UE \textit{SINR} (incl. considered the sum of \textit{SE}~\eqref{f:SE}) is achieved by an algorithm with a proper diagonal regularization, and ARZF is a suboptimal heuristic of such form. 
\end{Remark}

The possible interpretation of the function $J_{\bm S}(\bm W)$ (\ref{J}) is as follows. The second term $\lambda \| \bm W \|^2$ is the standard noise regularization part and the first term, the norm $\| \bm S ( \bm V \bm W - \bm I ) \|^2$, weighted by the matrix $\bm S$, weights more for the layers with higher singular values. And, therefore, the function is optimized more precisely for the layers with a higher signal quality compared to the layers with lower signal quality. In other words, precoding vectors for layers with higher singular values become similar to Zero-Forcing precoding, and for layers with lower singular values become similar to Maximum-Ratio Transmission, i.e. {\it ARZF provides adaptive regularization}. In the next section, we will see that this approach leads to a uniform increase in spectral efficiency compared to the basic method with unit weights.

Let us study the relation of ARZF with other algorithms. 
Firstly, we see that regularization parameter in WRZF is the (arithmetical) average of ARZF regularization:
$$\frac{\sigma^2}{P} {\rm tr} (\bm S^{-2}) = \frac{\sigma^2 L}{P} \cdot \frac 1 L \sum_{l=1}^L s_l^{-2}.$$
In the case when path-losses of all UE are similar $s_l \approx s, l=1,...,L$ ARZF and WRZF provide a similar result. Secondly, the relation between $\bm W_{RZF}$ and $\bm W_{ARZF}$ precoding is formulated as follows:
\begin{Theorem}
Assume the channel matrix has the form of $\bm H = \bm U^ \mathrm H \bm S \bm V$ (Lemma~\ref{lem_decomposition}) and denote $\bm F = \bm U \bm H = \bm S \bm V$. Then, the following relation for ARZF holds:
\begin{equation}\label{ARZF-RZF}
\bm W_{ARZF}(\bm V) = \bm W_{RZF}(\bm F) \bm S.
\end{equation}
\end{Theorem} 
\begin{proof}
\begin{multline}
 \bm W_{RZF}(\bm F) \bm S  = \bm F^ \mathrm H(\bm F \bm F^ \mathrm H + \lambda \bm I)^{-1} \bm S = \bm V^ \mathrm H \bm S (\underbrace{\bm S \bm V \bm V^ \mathrm H \bm S + \lambda \bm I}_{\bm B})^{-1} \bm S  \\ = \bm V^ \mathrm H \bm S \bm S^{-1} (\bm V \bm V^ \mathrm H + \bm S^{-1} \lambda \bm I \bm S^{-1})^{-1} \bm S^{-1} \bm S  =  \bm V^ \mathrm H(\underbrace{\bm V \bm V^ \mathrm H + \lambda \bm S^{-2} }_{\bm A})^{-1} = \bm W_{ARZF}(\bm V).
\end{multline}
\end{proof}
Right factor $\bm S$ in the r.h.s. of~\eqref{ARZF-RZF} can be interpreted as a special type of \emph{Power Allocation} algorithm (see an interesting study in~\cite[sec.~7]{Bjornson_tb_17}) that distributes the total transmission power between layers.
In practice, it is better to use $\bm W_{RZF}(\bm V)$ rather than $\bm W_{RZF}(\bm F)$, because the norms of the rows of the matrix $\bm F \bm F^ \mathrm H + \lambda \bm I$ can differ sufficiently (up to several orders!), which leads to unbalanced power distribution between layers (as state Theorem~\ref{theorem_ARZF} another way is to apply a proper Power allocation for $\bm W_{RZF}(\bm F)$. On the other hand, the regularization parameter of $\bm W_{RZF}(\bm F)$ is more natural and correct. The proposed $\bm W_{ARZF}(\bm V)$ combines the benefits from these two approaches and provides an envelope of them.

\begin{Remark} ARZF formula~\eqref{ARZF} can be found using the PCA-decomposition~\cite{PCA}, which is stated in the Theorem~\ref{theorem_ARZF}.
\end{Remark}

In this work we prove theoretical estimates of the number of conditionality of the inverse covariance matrix of the ARZF method and the standard RZF method, which is important for systems with fixed computational accuracy

\begin{Theorem}\label{prop:arzf_cond}
    Let \( \bm V^\mathrm{H} \bm V (\varepsilon) = \bm I + O (\varepsilon), \; \varepsilon \rightarrow 0 \), and given matrices \(\bm A = \bm V^\mathrm{H} \bm V (\varepsilon) + \lambda \bm S^{-2} \rightarrow \bm I + \lambda \bm S^{-2}  \) and \( \bm B = \bm S \bm V^\mathrm{H} \bm V (\varepsilon) \bm S + \lambda \bm I \rightarrow \bm S^2 + \lambda \bm I \), that are inverted in the corresponding precodings \(\bm W_{ARZF} = \bm V^\mathrm{H} \bm A^{-1} \), and \( \bm W_{RZF}= \bm S \bm V^\mathrm{H} \bm B^{-1} \). Then the conditioning numbers of $\bm A$ and $\bm B$ matrices are equal, respectively: \begin{equation}\label{eq:chis}
        \chi(\bm A) = \dfrac{\lambda s^{-2}_{\min} + 1 }{\lambda s^{-2}_{\max} + 1} \text{ and } \chi(\bm B) = \dfrac{\lambda + s^2_{\max}}{\lambda + s^2_{\min}}
    \end{equation}
    and are also related by the ratio:
    \begin{enumerate}
        \item \( \chi(\bm A) < \chi (\bm B) \), if \( \lambda < s^2_{\min} < s^2_{\max} \),
        \item \( \chi(\bm A) > \chi (\bm B) \), if \( s^2_{\min} < s^2_{\max} < \lambda \),
    \end{enumerate}
    where $\lambda = \frac{\sigma^2 L}{P}$

\end{Theorem}

\begin{proof}
The assumption that the matrix of singular user vectors is close to unitary is valid under low correlation user selection: \( \bm V^\mathrm{H} \bm V (\varepsilon) = \bm I + O (\varepsilon) \), with $\varepsilon \rightarrow 0$. Then the matrices under study are reduced to diagonal: $$\bm A = \bm V^\mathrm{H} \bm V (\varepsilon) + \lambda \bm S^{-2} = \bm I + O (\varepsilon) + \lambda \bm S^{-2} \rightarrow \bm I + \lambda \bm S^{-2}$$$$\bm B = \bm S \bm V^\mathrm{H} \bm V (\varepsilon) \bm S + \lambda \bm I = \bm S ( \bm I + O (\varepsilon) ) \bm S + \lambda \bm I = \bm S^2 + \lambda \bm I + O (\varepsilon) \bm S^2 \rightarrow \bm S^2 + \lambda \bm I $$

Their conditionality functions are equal~\eqref{eq:chis}.

By comparing these functions, we get transition points where it is clearly better to use the first formula \( \lambda < s^2_{\min} < s^2_{\max} \): \( \chi(\bm A) < \chi (\bm B) \), it is better to use the second \( s^2_{\min} < s^2_{\max} < \lambda \): \( \chi(\bm A) > \chi (\bm B) \).
\end{proof}

\begin{Remark}
In the situation \( s^2_{\min} < \lambda < s^2_{\max} \) nothing can be said and further investigation is required. Note that only case 1) \( \lambda < s^2_{\min} < s^2_{\max} \) is used in real networks. Singular vectors whose singular values are smaller than the noise power are not used, so the condition \( \chi(\bm A) < \chi (\bm B) \) is always satisfied.
\end{Remark}


The mathematical notation \(\bm W_{ARZF} = \bm V^\mathrm{H} \bm A^{-1} \) improves the conditionality of the system under low to medium noise $\lambda$ conditions, which makes the algorithm computationally more accurate under limited discharge grid conditions compared to another mathematical notation of the method \(\bm W_{RZF} = \bm S \bm V^\mathrm{H} \bm B^{-1} \). An experimental comparison of the conditionality value is shown in Fig.~\ref{fig:condition_number}. 

\begin{figure}
    \includegraphics[width=\linewidth]{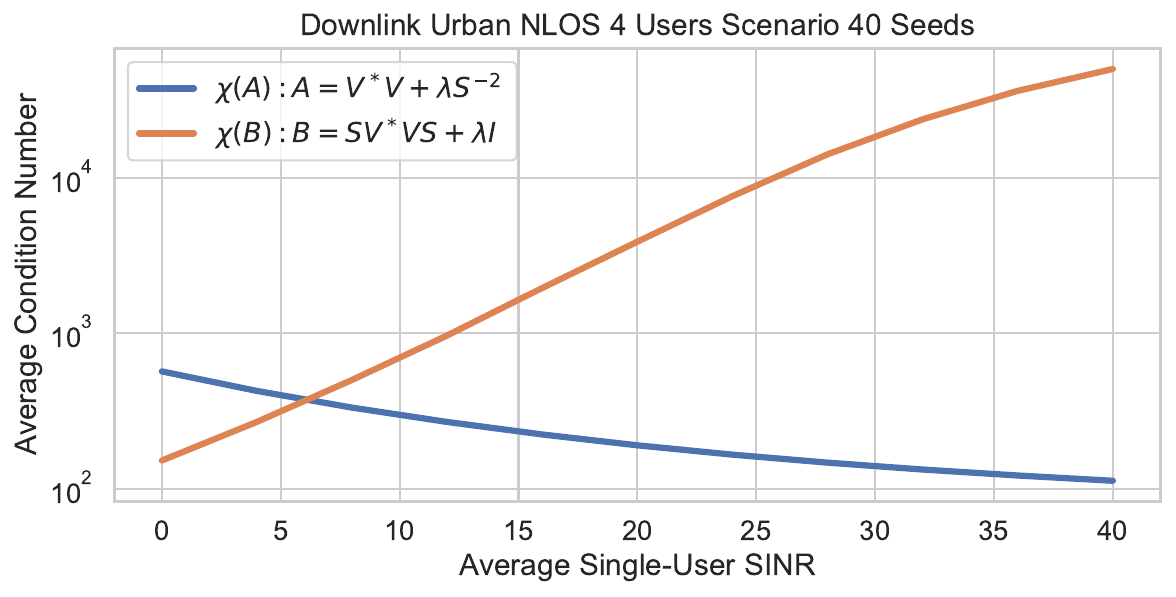}
    \caption{Conditionality of $\chi(\bm A)$ and $\chi(\bm B)$ matrices $\bm A$ and $\bm B$.} 
    \label{fig:condition_number}
\end{figure}

\subsubsection{Asymptotic properties of \textit{ARZF}}
Using the Neumann series~\cite{Neumann}, we formulate the following
\begin{Lemma}\label{lem_as}
Consider square Hermit matrices $\bm A$ and $\bm B$ of the same size and full rank. Introduce a small parameter $0<\varepsilon\ll 1$. The following  and same size square matrices the following asymptotic holds: 
\begin{equation}
(\bm A + \varepsilon \bm B)^{-1} = \bm A^{-1} (\bm I - \varepsilon \bm B \bm A^{-1} + \Xi_1) = (\bm I - \varepsilon \bm A^{-1} \bm B  + \Xi_2) \bm A^{-1},
\end{equation}
where $\|\Xi_1\| \sim \|\Xi_2\| \sim \varepsilon^2 \|\bm B\|^2 \|\bm A\|^{-2}$ as $\varepsilon \rightarrow 0$.
\end{Lemma}
\begin{proof}
We are looking for the inverse in the form of formal series $(\bm A + \varepsilon \bm B)^{-1} = \bm A_0 + \varepsilon \bm A_1 + \dots$. Calculating product of $\bm A + \varepsilon \bm B$ and its inverse, we get the chain of equations for each power of $\varepsilon$:
\begin{align}&
\bm A \bm A_0 = \bm A_0 \bm A = \bm I \quad &\Rightarrow& \quad \bm A_0 = \bm A^{-1},
\\ & 
\varepsilon(\bm B \bm A_0 + \bm A \bm A_1) = \varepsilon (\bm A_0 \bm B + \bm A_1 \bm A) = 0 \quad &\Rightarrow& \quad \bm A_1 = - \bm A^{-1} \bm B \bm A^{-1}, \quad \dots
\end{align}
\end{proof}
The following theorem~\ref{theorem_as} gives the asymptotic properties of ARZF.
\begin{Theorem}\label{theorem_as}
Assume that matrix $\bm V \bm V^ \mathrm H$ 
is of full rank equals to $L$. It holds:
\begin{align}&
\bm W_{ARZF}(\bm V) = \bm W_{ZF}(\bm V) + O(\lambda), &\quad& {\rm as} \;\; \lambda\rightarrow 0+,
\\&
\bm W_{ARZF}(\bm V) = \bm W_{MRT}(\bm V) \bm S^2 + O(\lambda^{-1}), &\quad& {\rm as} \;\; \lambda\rightarrow +\infty.
\end{align}
\end{Theorem}
\begin{proof}
To get asymptotics as $\lambda\to0+$ one should apply Lemma~\ref{lem_as} with $\bm A = \bm V \bm V^ \mathrm H$ and $\bm B = \bm S^{-2}$, asymptotics as $\lambda\to +\infty$ are given by Lemma~\ref{lem_as} with $\bm A = \bm S^{-2}, \; \bm B = \bm V \bm V^ \mathrm H, \; \varepsilon = \lambda^{-1}$. Also note that normalizing coefficient $\mu$~\eqref{mu} is different for different precoding algorithms. 
\end{proof}

Property at low noise $\lambda\to0+$ is the same as for $\bm W_{RZF}$. Another asymptotic from Theorem~\ref{theorem_as} means that, when noise is greater than signal even for UE with the best channel, ARZF only serves UE with the best channel.

\subsubsection{Gradient-Based Optimal Regularization (OPT)}

\begin{algorithm}
\SetAlgoLined
\KwIn{Channel $\bm H$ and its decomposition $\bm H = \bm U^ \mathrm H \bm S \bm V$ by Lemma~\ref{lem_decomposition}, station power $P$, noise $\sigma^2$, \# of iterations $N$;}
\textbf{Define}  precoding function $\bm W(\bm R)$ using~\eqref{Diff RZF}; \\
\textbf{Define}  detection function $\bm G(\bm R) = \bm G^{\textit{MMSE}} (\bm W (\bm R) )$ using~\eqref{MMSE_1};
\\
\textbf{Define}  target function $J^{\textit{SE}}(\bm R) = \textit{SE}(\bm W(\bm R), \bm H, \bm G(\bm R), \sigma^2)$ of
\eqref{Symbol SINR},~\eqref{f:sek},~\eqref{J_SE};
\\
\textbf{Set} tolerance grad $\varepsilon_g$, termination tolerance on first order optimality (1e-5);
\\
\textbf{Set} tolerance change $\varepsilon_c$, termination tolerance on function and parameters (1e-9);
\\
\textbf{Initialize} regularization by $\bm R_0 = \dfrac{L\sigma^2}{P} \bm S^{-2}$; 
\\
\For{$n=0$ \KwTo $N-1$}{

Calculate gradient $\nabla_{\bm R} J^{\textit{SE}}(\bm R)|_{\bm R = \bm R_n}$ \;

\eIf{Convergence conditions hold on $\varepsilon_g$ or $\varepsilon_c$}{
     \Return{$\bm W_\textit{OPT} = \bm W (\bm R_{n})$}}{

Calculate L-BFGS direction~\cite{L-BFGS} using the gradient: $\bm D_n = \bm D (\nabla_{\bm R} J^{\textit{SE}}(\bm R)|_{\bm R = \bm R_n})$ \;
Optimize scalar step length $\alpha_n = \arg\max\limits_\alpha J^{\textit{SE}}(\bm R_n + \alpha \bm D_n)$\;
Make the optimization step: $\bm R_{n+1} = \bm R_n + \alpha_n \bm D_n$\;
}}
\Return{$\bm W_\textit{OPT} = \bm W (\bm R_N)$}
\caption{On the Optimal Precoding Regularization $\bm R$}
\label{OPT Scheme}
\end{algorithm}%

The proposed algorithm provides minimum to quadratic optimization problem~\eqref{prob_quadratic_S}, but still, there is a question: how good is it w.r.t. to sum \textit{SE} function~\eqref{J_SE}?
In~\cite{Bjornson} it is proved that optimal solution to the function $f(\textit{SINR}_1,...,\textit{SINR}_K)$ with total power constraint is given by algorithm of the form
\begin{equation}
\bm W_\textit{OPT}(\bm V) = \mu \bm V^ \mathrm H (\bm V \bm V^ \mathrm H + \bm R)^{-1} \bm P, \quad \bm R = {\rm diag}(r_1,\dots, r_L), \quad \bm P = {\rm diag}(p_1,\dots, p_L).
\end{equation}
It is hardly possible to explicitly express the particular values of $\bm R, \bm P$ for optimal precoding but the structure itself is useful.   
According to this result, we are going to study the effectiveness of the proposed {ARZF} algorithm comparing it with gradient search over the sum of \textit{SE}~\eqref{f:SE}. 

To this end we formulate 
an iterative solution with differentiable embedded parts. The proposed parametric solution preserves the structure of the basic {RZF} algorithm and optimizes the target functional of \textit{SE}, which leads to a significant improvement in quality.  We set the constrained smooth optimization problem. The problem is to find the local maximum of the \textit{Spectral Efficiency}~\eqref{J_SE}:
\begin{maxi}|l|
  {\text{diag}\{r_1 \dots r_L\} = \bm R}{\textit{SE}(\bm W (\bm R), \bm H, \bm G^{\textit{MMSE}}(\bm W), \sigma^2)}{}{}, \quad \text{s.t.:}\;\; \|(w_{t1},\dots,w_{tL})\|^2 \leqslant \frac P T
\end{maxi}



A parametric solution uses the RZF formula as follows:
\begin{align} & \label{Diff RZF}
    \bm W (\bm V, \bm R) = \mu (\widehat{\bm W}) \widehat{\bm W} (\bm V, \bm R),\quad \widehat{\bm W}(\bm R) =  \bm V^ \mathrm H (\bm V \bm V^ \mathrm H + \bm R)^{-1}, \quad
        \\ & \nonumber
    \mu (\widehat{\bm W}) = \frac{\sqrt{P / T}}{ \max\limits_t \{\|(w_{t1},\dots,w_{tL})\| \}^T_{m=1}}.
\end{align}

We restrict the maximum power of antennas by multiplying the precoding matrix by the scalar, which allows us to satisfy the power constraints and save the geometry and desired properties of the constructed precoding.

In further experiments, we will make the real diagonal matrix $\bm R \in \mathbb R ^ {L \times L}$ differentiable and optimize it for the target Spectral Efficiency functional, which is one of our contributions. The optimization of precoding matrix $\bm W$ is given in the corresponding article~\cite{QNS}.

Detection is involved in the calculation of the gradient, it can be considered as an integral part of the Spectral Efficiency. That is, the differentiable variables here are the diagonal of the regularization matrix, then the precoding matrix is calculated using the regularization, then this precoding is substituted into the detection. Finally, the calculated precoding and detection are both substituted into the gradient. The regularization diagonal is involved in all these operations as an internal variable of a composite function, and its gradient can be calculated using a chain rule, back-propagation algorithm, as it happens using automatic differentiation of the PyTorch library.

\begin{Remark}
In considered gradient optimization precoding is constructed, assuming some particular detection that could be performed on the UE side (namely, MMSE detection). This idea is widely discussed in state-of-the-art (see e.g.~\cite{PrecodingDetection}) and can be used to improve merely any precoding policy with a corresponding iterative procedure.
\end{Remark}

\begin{figure}
    \centering
    \includegraphics[width=0.8\linewidth]{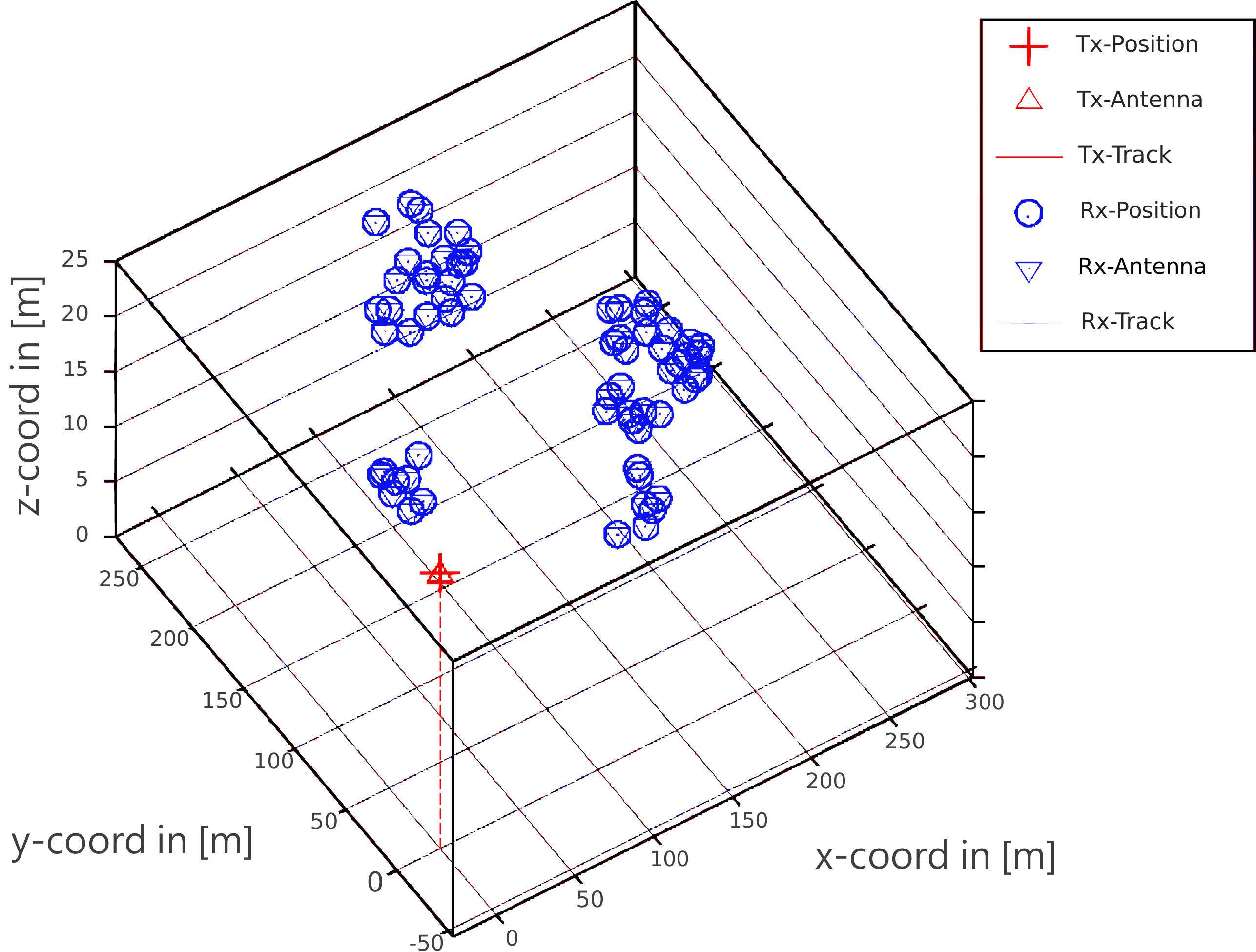}
    \caption{An example of random user generation in an Urban environment.}
    \label{fig:quadriga_layout}
\end{figure}

\section{Simulation Results and Discussion}

\subsection{Setup in Quadriga}

In this section, we describe how we obtain data from Quadriga, open-source software for generating a realistic radio channel. The considered scenario is Urban Non-Line-of-Sight 3GPP\_38.901\_RMa\_NLOS~\cite{Tse_tb_05, NLOS}.

Fig.~\ref{fig:quadriga_layout} shows an example of user placement, and users are assigned either in a cluster of one of the buildings or on the ground next to a building. The users are highlighted with blue circles, and the base station~--- one in red. The distances between the individual antennas of the station and the users are negligible compared to the distances between the users and the station, so the station and users are depicted as separate circles, each containing multiple antennas within it.

The overall procedure is as follows, for each random seed:
\begin{enumerate}
    \item We generate a random environment around the base station;
    \item We select random user positions near the base station;
    \item We select relevant users based on correlations.
\end{enumerate}

Next, we describe the procedure in detail.

First, we fix the position of the base station at the coordinates $[0, 0, 25]$ and set the random seed in Quadriga~\cite{Quadriga} to generate a random environment. 

Second, we select the random user positions around the base station. To do this, we locate users in the urban landscape (see example at  Fig~\ref{fig:quadriga_layout}):
\begin{enumerate}
    \item We sample up to 8 cluster centers $x_c$, $y_c$ in the 120$^{\circ}$~sector from the base station within 2000m from the base station. Each cluster represents a part of a city building;
    \item We assign a random cluster height $z_c = 1.5$m $+ (3 \cdot U(\{1, \ldots, 10\}) - 1)$ , selecting the cluster floor in a building from the uniform distribution $U$;
    \item For each user, we assign a cluster id $c(u)$ and sample $x_u, y_u$ position randomly over the 60m circle around the cluster centre;
    \item We sample the height of each user in addition to the height of the cluster, at the same time, 80\% of users are placed near the cluster floor $ z_u = z_{c(u)} + 3 \cdot U(\{-1, 0, 1\})$ m. and 20\% of users are placed outdoors $z_u = 1.5$m.
\end{enumerate} 

Third, after generating channel matrices for a fairly large number of users $K_{\max} = 64$, we perform the user selection to choose subsets of users that will be served together (this procedure simulates Scheduler). We select one subset of $K < K_{\max}$ users, such that there is no pair of too high correlated users (in real life users with high correlations are served at a different time or on different frequency intervals). The correlation between users $i, j$ is  measured as the squared cosine between the main singular vectors: ${\rm corr}_{i,j} = |\bm v_{i,1}^ \mathrm H \bm v_{j,1}|^2 \leqslant 0.3$. The number of transmitted symbols for each UE is $L_k = 2$ (it is the simplest Rank Selection policy).

We also consider two different situations (scenarios):
\begin{enumerate}
    \item When UE have different path-losses (PL), when we add random factor \\ $\text{PL}\in [-10 \textrm{dB}, 10 \textrm{dB}]$, which is the realistic channel variance for close UE (in real networks path-loss variance for UE within one cell can be up to $60 \textrm{dB}$).
    \item When UE have almost equal path-losses, i.e. first singular values of each two UE are the same $s_{i,1} \sim s_{j,1}$ (this is the default option for Quadriga channel generation);
\end{enumerate}

Presented results are the average over 40 considered random realizations of channel matrix and UE subsets. 
Namely, for each scenario we generate 40 different channels: $\bm H \in \mathbb{C}^{K \times R_k \times T}$:
\begin{itemize}
    \item $T = 64$ base station antennas;
    \item $K = 4$ users;
    \item $R = 16$ total user antennas;
    \item $L = 8$ total user layers.
\end{itemize}
The carrier frequency for each channel matrix is selected randomly over the bandwidth. The base station antenna array forms a grid with $8$ placeholders along the $y$ axis and $4$ placeholders along the $y$ axis, the receiver antenna array consists of two placeholders along the $x$. Each placeholder contains two cross-polarized antennas. Our user generation algorithm produces realistic setups for the Urban case and can be convenient 
for other studies. An interested reader can find detailed hyperparameters for antenna models and generation processes in Table~\ref{tab:quadriga_hyp} in Appendix.

\subsection{Results and Discussion}\label{sec_simulation}

We compare the proposed algorithm (ARZF) with reference algorithms (MR, ZF, RZF, WRZF) and the estimation of the precoding with optimal diagonal regularization (OPT).
In Fig.~\ref{Avg SE Pathloss} (Tab.~\ref{tab:Avg SE Pathloss}) and Fig.~\ref{Min SE Pathloss} (Tab.~\ref{tab:Min SE Pathloss}) 
the Average Spectral Efficiency~\eqref{J_SE} and the Minimal SE~\eqref{Min SE} of the described algorithms are presented for the different path-loss scenarios. In Fig.~\ref{Avg SE Casual}  (Tab.~\ref{tab:Avg SE Pathloss}) and Fig.~\ref{Min SE Casual} (Tab.~\ref{tab:Min SE Casual}) the same quality functions are compared for the equal path-loss scenario. 

ARZF (the red line,~\ref{ARZF}) provides the best Average Spectral Efficiency~\eqref{J_SE} compared to all other analytical methods up to the highest \textit{SUSINR} region (see Fig.~\ref{Avg SE Pathloss},~\ref{Avg SE Casual}). The advantages of the ARZF algorithm are revealed due to adaptive regularization for a specific path-loss (and thus singular values order) for each user. In detail, the situation is as follows. ZF ($\bm W_{ZF}(\bm F), \bm F = \bm S \bm V$ and $\bm W_{ZF}(\bm V)$) are better than MRT ($\bm W_{MRT}(\bm V)$) on high \textit{SUSINR} region, and are worse for low  \textit{SUSINR} ($\bm W_{ZF}(\bm V)$ is always better than $\bm W_{ZF}(\bm F)$). RZF is better than both MRT and corresponding ZF, e.g. $\bm W_{RZF}(\bm V)$ is better than $\bm W_{MRT}(\bm V)$ and $\bm W_{ZF}(\bm V)$. One can say that RZF provides an envelope of MRT and ZF. It is worth saying that although $\bm W_{RZF}(\bm V)$ outperforms $\bm W_{RZF}(\bm F)$ on high \textit{SUSINR} region $\bm W_{RZF}(\bm F)$ is better for low \textit{SUSINR}. Due to the adaptive approach, ARZF makes an envelope of both RZF algorithms. It is quite natural that on equal path-loss scenario, WRZF and ARZF work similarly: it happens because all the singular values are almost the same. Well, on the different path-loss scenarios, WRZF degrades significantly and ARZF is better. 

One may notice, that the gradient-based iterative method OPT (the black line,~\eqref{OPT Scheme}) provides superior results. On the other hand, OPT is computationally expensive and cannot be used in practice. However, OPT is very good for upper bound estimation: the results for OPT show that the ARZF method can be improved in further research. We also measure the quality of Minimal Spectral Efficiency~\eqref{Min SE} (see Fig.~\ref{Min SE Pathloss},~\ref{Min SE Casual}) to investigate the performance of the weakest user.

These results show that the ARZF method, configuring the regularization in a special adaptive way, outperforms the Average SE at the expense of Minimal SE, i.e. of the weaker users. This follows from the fact that in terms of the Minimal SE function, the ARZF method can lose compared to other algorithms, especially in the low \textit{SUSINR} region. The Average SE quality of the system, however, increases significantly. The contradiction between Average and Minimal SE functions is well-known (see. e.g.~\cite{Bjornson_tb_17}).

Fig.~\ref{Avg SE Casual} (Tab.~\ref{tab:Avg SE Casual}) and Fig.~\ref{Min SE Casual} (Tab.~\ref{tab:Min SE Casual}) represent Average and Minimal scenarios and show the experimental results for users with almost equal path-losses. The general tendency is the same as for different path-loss scenarios: The proposed ARZF method outperforms all other analytical references in the Average SE, but it is not the best from the Minimal SE point of view. In this scenario the gain in Average SE is not so big, particularly, ARZF behaves almost the same way as WRZF. The last observation in Fig~\eqref{Avg SE Casual} (Tab.~\eqref{tab:Avg SE Casual}) is very important. The red line of $W_{ARZF}(\bm V)$ coincides with the blue line of $W_{WRZF}(\bm V)$. This result shows that both algorithms work the same when users have equal path-losses, which is also confirmed by theoretical calculations


\section{Conclusions}\label{sec_conclusion}
Multi-user precoding optimization is a key problem in modern cellular wireless systems, which are based on massive MIMO technology. In this paper, we analyze the performance of different transmission precoding techniques in a downlink multi-user scenario. Linear techniques are computationally less expensive. On the other hand, non-linear techniques can provide better performance. The first technique that we propose in this paper is a low-complexity heuristic formula of the Adaptive Regularized Zero-Forcing (ARZF) algorithm. This technique is especially attractive in cases when the users and the base station are equipped with multiple antennas and have various path losses. We study analytically the relation of ARZF with known RZF-like algorithms and its asymptotic properties. Finally, we study the properties of the proposed ARZF method on simulations, using a realistic Quadriga channel. Simulations show uniform improvement over the reference methods to \textit{Average Spectral Efficiency} for all considered scenarios. This particularly means that weighted MSE problem statement~\eqref{prob_quadratic_S} is a more adequate approximation of~\eqref{prob_phys} than the standard MMSE statement. \textit{Minimal Spectral Efficiency} function is not the best, but is still acceptable over ARZF. We also introduce a non-linear technique Gradient-Based Optimal Regularization (OPT) that performs gradient optimization on the target \textit{Spectral Efficiency} function and finds the optimal diagonal regularization matrix this way. This algorithm can be used for the upper bound study.

\section*{Declarations}

\subsection{Availability of data and materials}
The Quadriga-generated channel matrices, as well as the full Python implementation of the proposed approach and baselines, are publicly available in the Github repository at \url{https://github.com/eugenbobrov/Adaptive-Regularized-Zero-Forcing-Beamforming-in-Massive-MIMO-with-Multi-Antenna-Users}

\subsection{Authors Contribution }
EB and BC formulated the proposed formula of ARZF; EB, BC, DM and DY formulated, proved, and wrote the theoretical results section; EB did the main work of conducting the experiments and drafting the manuscript; BC also did the main work of the literature review and the introduction; EB and DM conceived the study; ST helped with dataset preparation; VK and DM were involved in methodology; HL and DZ contributed to project administration and supervision. All authors read and approved the final manuscript.

\subsection{Acknowledges}
Authors are grateful to E.~Barinova and E.~Levitskaya for support. 

\subsection{Funding}
Authors have received research support from Huawei Technologies. 

\subsection{Competing Interest}
The authors declare that they have no competing interests.

\begin{figure}
    \centering
    \includegraphics[width=1\linewidth]{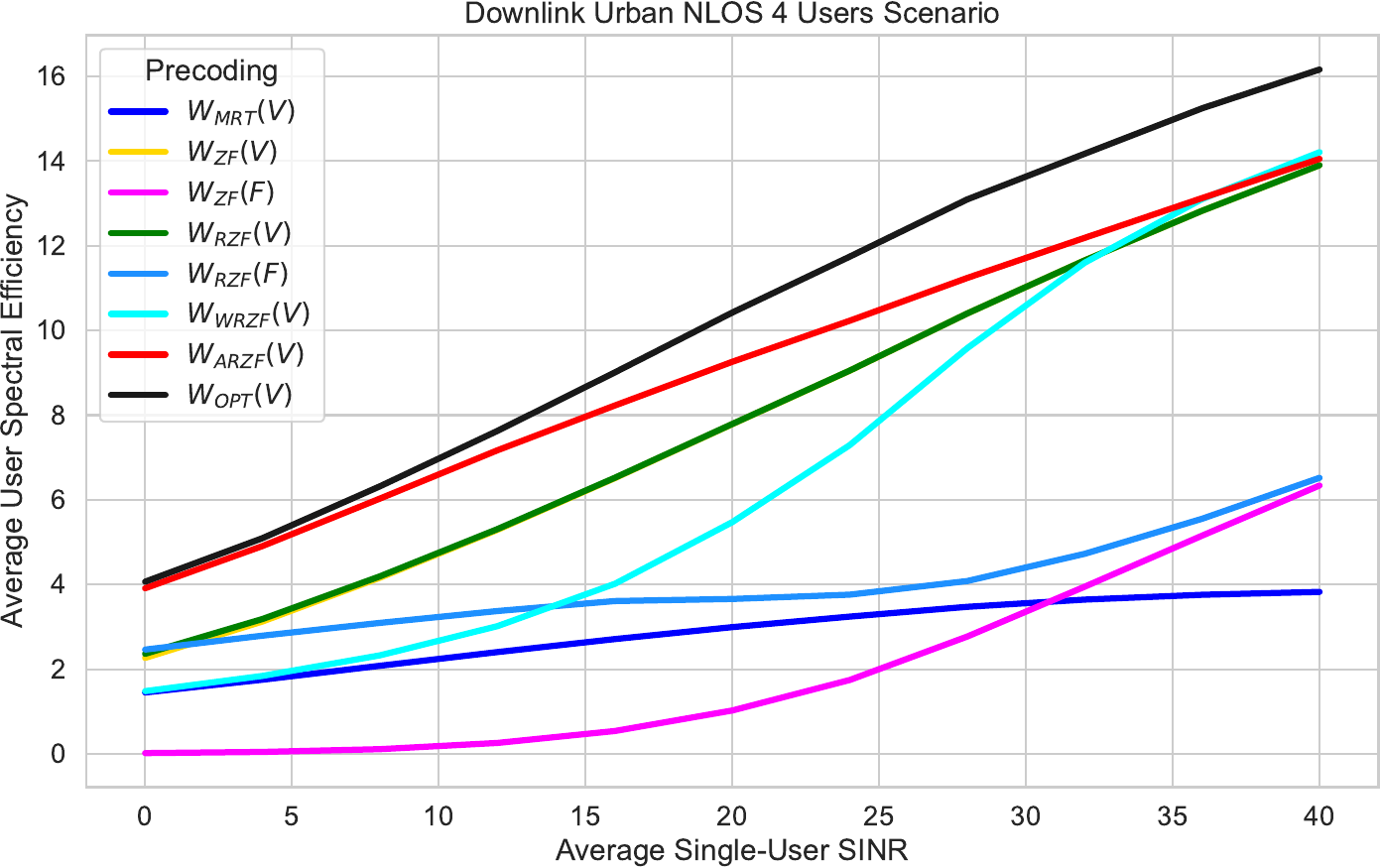}
    \caption{Average SE of the different Precoding algorithms in the \textit{Urban NLOS} scenarios \textit{using different path-losses}. The green line coincides with the yellow one. Matrix $\bm F = \bm S \bm V$. Values are presented in Tab.~\ref{tab:Avg SE Pathloss}.}
    \label{Avg SE Pathloss}
\end{figure}

\begin{figure}
    \centering
    \includegraphics[width=1\linewidth]{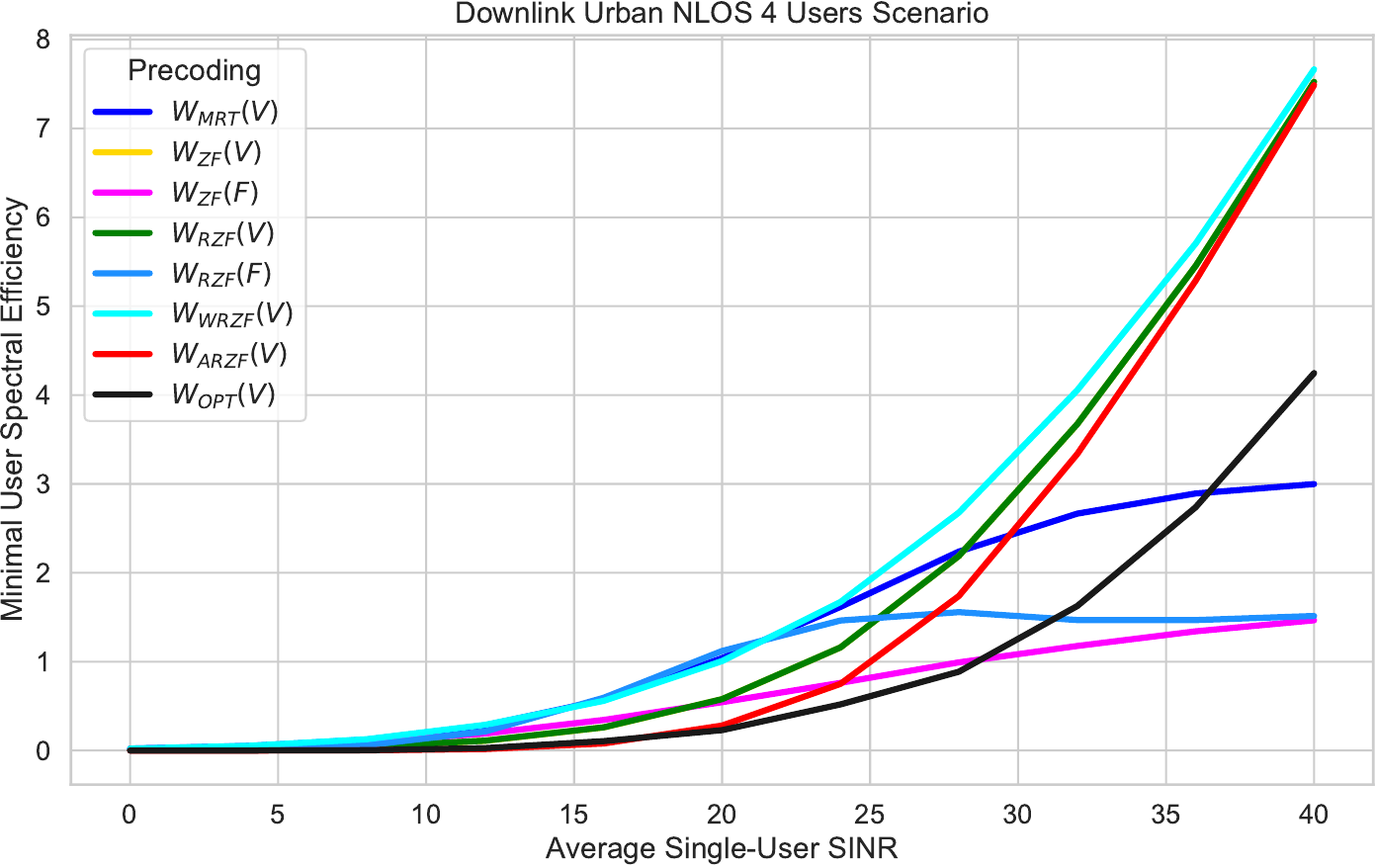}
    \caption{Minimal User SE of the different Precoding algorithms in the \textit{Urban NLOS} scenarios \textit{using different path-losses}. The green line coincides with the yellow one. Matrix $\bm F = \bm S \bm V$. Values are presented in Tab.~\ref{tab:Min SE Pathloss}.}
    \label{Min SE Pathloss}
\end{figure}

\begin{figure}
    \centering
    \includegraphics[width=1\linewidth]{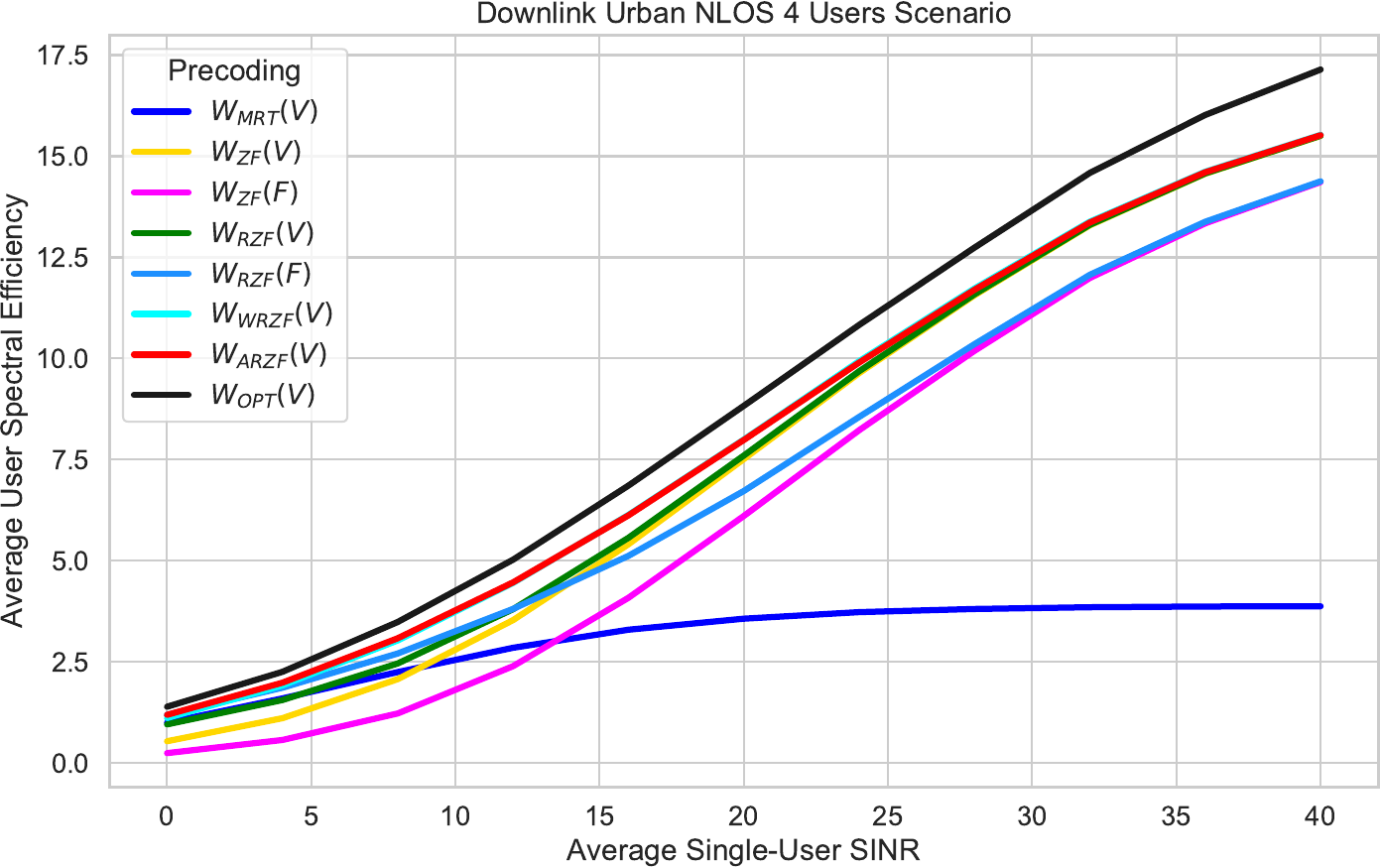}
    \caption{Average SE of the different Precoding algorithms in the \textit{Urban NLOS} scenario (Tab.~\ref{tab:Avg SE Casual}). The red line of $\bm W_{ARZF}(\bm V)$ coincides with the blue line of $\bm W_{WRZF}(\bm V)$. This result shows that the algorithms work the same when users have equal path-losses.}
    \label{Avg SE Casual}
\end{figure}

\begin{figure}
    \centering
    \includegraphics[width=1\linewidth]{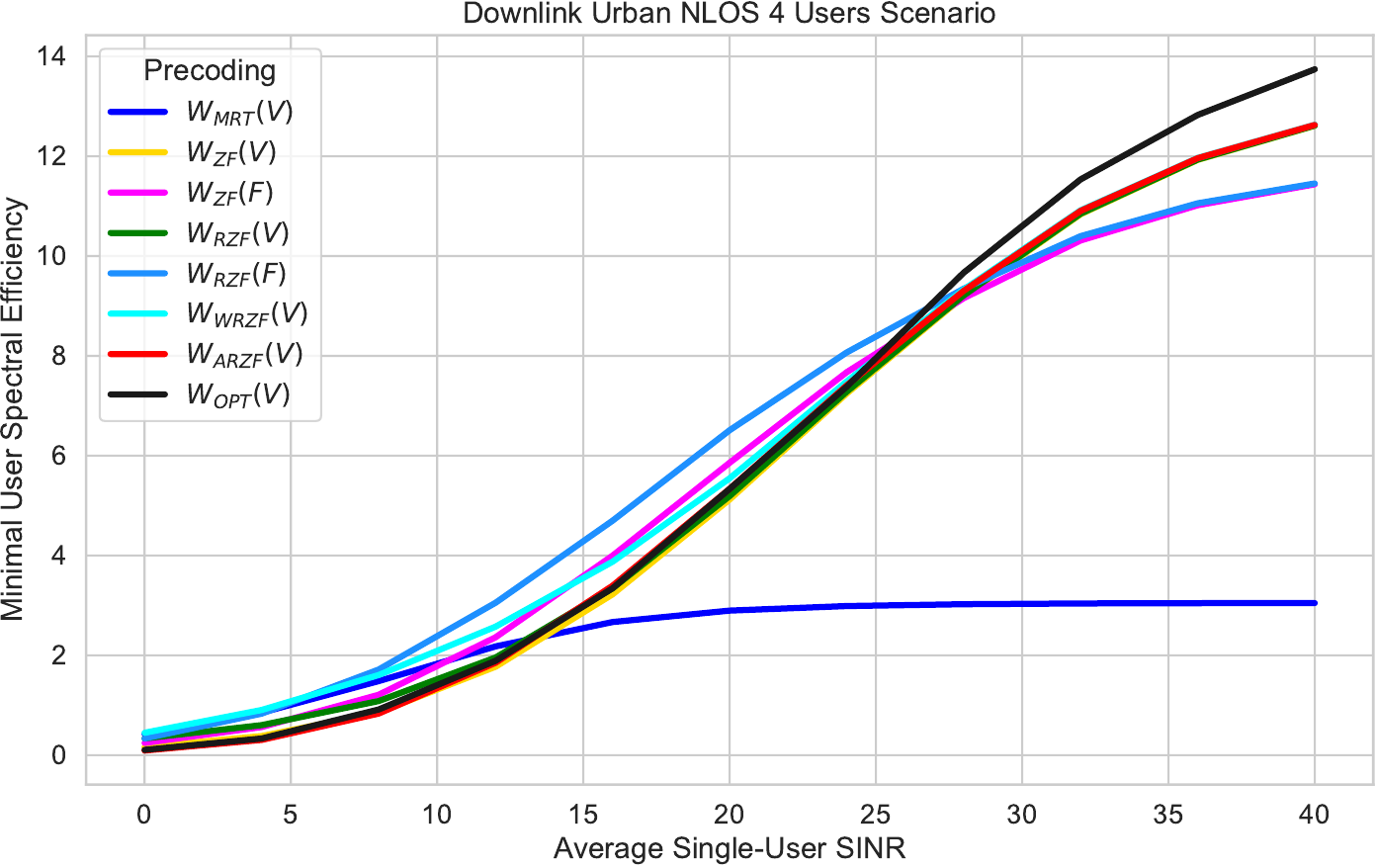}
    \caption{Minimal User SE of the different Precoding algorithms in the \textit{Urban NLOS} scenario (Tab.~\ref{tab:Min SE Casual}). }
    \label{Min SE Casual}
\end{figure}

\begin{table}
    \centering
    \resizebox{\columnwidth}{!}{%
    \begin{tabular}{l|r|r|r|r|r|r|r|r|r}
    \toprule
    Precoding &  $W_{MRT}(V)$ &  $W_{ZF}(V)$ &  $W_{ZF}(F)$ &  $W_{RZF}(V)$ &  $W_{RZF}(F)$ &  $W_{WRZF}(V)$ &  $W_{ARZF}(V)$ &  $W_\textit{OPT}(V)$ \\
    SU SINR &               &              &              &               &               &                &                &               \\
    \midrule
    0       &          1.45 &         2.27 &         0.02 &          2.37 &          2.46 &           1.48 &           3.91 &          4.07 \\
    4       &          1.75 &         3.13 &         0.05 &          3.19 &          2.79 &           1.84 &           4.91 &          5.10 \\
    8       &          2.08 &         4.16 &         0.11 &          4.19 &          3.10 &           2.33 &           6.03 &          6.32 \\
    12      &          2.41 &         5.29 &         0.26 &          5.31 &          3.37 &           3.02 &           7.17 &          7.63 \\
    16      &          2.71 &         6.51 &         0.54 &          6.52 &          3.61 &           4.02 &           8.23 &          9.01 \\
    20      &          2.99 &         7.79 &         1.03 &          7.79 &          3.66 &           5.48 &           9.26 &         10.42 \\
    24      &          3.24 &         9.05 &         1.75 &          9.05 &          3.76 &           7.29 &          10.23 &         11.74 \\
    28      &          3.47 &        10.40 &         2.77 &         10.40 &          4.08 &           9.58 &          11.24 &         13.09 \\
    32      &          3.65 &        11.65 &         3.95 &         11.65 &          4.73 &          11.60 &          12.19 &         14.17 \\
    36      &          3.76 &        12.83 &         5.16 &         12.83 &          5.55 &          13.10 &          13.13 &         15.24 \\
    40      &          3.83 &        13.90 &         6.34 &         13.90 &          6.52 &          14.21 &          14.05 &         16.17 \\
    \bottomrule
    \end{tabular}%
    }
    \caption{The table refers to Fig.~\ref{Avg SE Pathloss} and shows the Average Spectral Efficiency of the different precodings in the \textit{Urban NLOS} scenario \textit{using different path-losses}. Proposal algorithm is $\bm W_{ARZF} (\bm V)$. Optimal regularization $\bm W_\textit{OPT}(\bm V)$ was configured using the L-BFGS optimization algorithm.}
    \label{tab:Avg SE Pathloss}
\end{table}

\begin{table}
    \centering
    \resizebox{\columnwidth}{!}{%
    \begin{tabular}{l|r|r|r|r|r|r|r|r|r}
    \toprule
    Precoding &  $W_{MRT}(V)$ &  $W_{ZF}(V)$ &  $W_{ZF}(F)$ &  $W_{RZF}(V)$ &  $W_{RZF}(F)$ &  $W_{WRZF}(V)$ &  $W_{ARZF}(V)$ &  $W_\textit{OPT}(V)$ \\
    SU SINR &               &              &              &               &               &                &                &               \\
    \midrule
0       &          0.02 &         0.01 &         0.02 &          0.01 &          0.00 &           0.02 &           0.00 &          0.00 \\
4       &          0.05 &         0.02 &         0.04 &          0.02 &          0.01 &           0.05 &           0.00 &          0.00 \\
8       &          0.12 &         0.04 &         0.10 &          0.04 &          0.06 &           0.13 &           0.00 &          0.00 \\
12      &          0.28 &         0.11 &         0.19 &          0.11 &          0.21 &           0.29 &           0.02 &          0.03 \\
16      &          0.57 &         0.26 &         0.34 &          0.26 &          0.59 &           0.56 &           0.08 &          0.11 \\
20      &          1.04 &         0.58 &         0.54 &          0.58 &          1.12 &           1.01 &           0.28 &          0.23 \\
24      &          1.62 &         1.16 &         0.76 &          1.16 &          1.46 &           1.67 &           0.75 &          0.52 \\
28      &          2.24 &         2.19 &         0.99 &          2.19 &          1.56 &           2.68 &           1.74 &          0.89 \\
32      &          2.67 &         3.67 &         1.18 &          3.67 &          1.47 &           4.05 &           3.34 &          1.63 \\
36      &          2.89 &         5.45 &         1.34 &          5.45 &          1.47 &           5.70 &           5.29 &          2.74 \\
40      &          3.00 &         7.52 &         1.47 &          7.52 &          1.51 &           7.67 &           7.49 &          4.25 \\
    \bottomrule
    \end{tabular}%
    }
    \caption{The table refers to Fig.~\ref{Min SE Pathloss} and shows the Minimal Spectral Efficiency of the different precodings in the \textit{Urban NLOS} scenario \textit{using different path-losses}. Proposal algorithm is $\bm W_{ARZF} (\bm V)$. Optimal regularization $\bm W_\textit{OPT}(\bm V)$ was configured using the L-BFGS optimization algorithm.}
    \label{tab:Min SE Pathloss}
\end{table}

\begin{table}
    \centering
    \resizebox{\columnwidth}{!}{%
    \begin{tabular}{l|r|r|r|r|r|r|r|r}
    \toprule
    Precoding &  $W_{MRT}(V)$ &  $W_{ZF}(V)$ &  $W_{ZF}(F)$ &  $W_{RZF}(V)$ &  $W_{RZF}(F)$ &  $W_{WRZF}(V)$ &  $W_{ARZF}(V)$ &  $W_\textit{OPT}(V)$ \\
    \textit{SUSINR} &               &              &              &               &               &                &                &               \\
    \midrule
    0       &          1.00 &         0.54 &         0.25 &          0.96 &          1.19 &           1.12 &           1.19 &          1.39 \\
    4       &          1.60 &         1.11 &         0.57 &          1.56 &          1.86 &           1.92 &           1.99 &          2.26 \\
    8       &          2.25 &         2.08 &         1.23 &          2.46 &          2.71 &           3.03 &           3.09 &          3.49 \\
    12      &          2.85 &         3.54 &         2.40 &          3.81 &          3.81 &           4.46 &           4.47 &          5.03 \\
    16      &          3.30 &         5.41 &         4.09 &          5.57 &          5.12 &           6.14 &           6.12 &          6.86 \\
    20      &          3.57 &         7.52 &         6.11 &          7.60 &          6.73 &           8.00 &           7.99 &          8.83 \\
    24      &          3.73 &         9.64 &         8.23 &          9.69 &          8.56 &           9.94 &           9.90 &         10.83 \\
    28      &          3.81 &        11.56 &        10.19 &         11.59 &         10.36 &          11.73 &          11.70 &         12.74 \\
    32      &          3.85 &        13.29 &        11.99 &         13.31 &         12.07 &          13.38 &          13.36 &         14.59 \\
    36      &          3.87 &        14.57 &        13.34 &         14.58 &         13.38 &          14.62 &          14.60 &         16.02 \\
    40      &          3.88 &        15.50 &        14.36 &         15.51 &         14.38 &          15.53 &          15.52 &         17.15 \\
    \bottomrule
    \end{tabular}%
    }
    \caption{The table refers to Fig.~\ref{Avg SE Casual} and shows the Average Spectral Efficiency of the different precodings in the \textit{Urban NLOS scenario}. Proposal algorithm is $\bm W_{ARZF} (\bm V)$. Optimal regularization $\bm W_\textit{OPT}(\bm V)$ was configured using the L-BFGS optimization algorithm.}
    \label{tab:Avg SE Casual}
\end{table}

\begin{table}
    \centering
    \resizebox{\columnwidth}{!}{%
    \begin{tabular}{l|r|r|r|r|r|r|r|r}
    \toprule
    Precoding &  $W_{MRT}(V)$ &  $W_{ZF}(V)$ &  $W_{ZF}(F)$ &  $W_{RZF}(V)$ &  $W_{RZF}(F)$ &  $W_{WRZF}(V)$ &  $W_{ARZF}(V)$ &  $W_\textit{OPT}(V)$ \\
    \textit{SUSINR} &               &              &              &               &               &                &                &               \\
    \midrule
    0                        &          0.42 &         0.16 &         0.24 &          0.33 &          0.33 &           0.45 &           0.09 &          0.11 \\
    4                        &          0.85 &         0.38 &         0.56 &          0.60 &          0.83 &           0.91 &           0.30 &          0.33 \\
    8                        &          1.48 &         0.86 &         1.21 &          1.08 &          1.71 &           1.60 &           0.83 &          0.91 \\
    12                       &          2.18 &         1.77 &         2.36 &          1.96 &          3.05 &           2.58 &           1.85 &          1.89 \\
    16                       &          2.67 &         3.22 &         4.00 &          3.34 &          4.70 &           3.88 &           3.40 &          3.35 \\
    20                       &          2.90 &         5.12 &         5.86 &          5.19 &          6.51 &           5.55 &           5.34 &          5.33 \\
    24                       &          2.99 &         7.25 &         7.67 &          7.28 &          8.07 &           7.49 &           7.41 &          7.38 \\
    28                       &          3.02 &         9.19 &         9.16 &          9.21 &          9.35 &           9.33 &           9.30 &          9.67 \\
    32                       &          3.04 &        10.84 &        10.31 &         10.85 &         10.40 &          10.92 &          10.91 &         11.54 \\
    36                       &          3.04 &        11.93 &        11.02 &         11.93 &         11.06 &          11.97 &          11.96 &         12.83 \\
    40                       &          3.05 &        12.61 &        11.44 &         12.62 &         11.45 &          12.63 &          12.63 &         13.74 \\
    \bottomrule
    \end{tabular}%
    }
    \caption{The table refers to Fig.~\ref{Min SE Casual} and shows the Minimal Spectral Efficiency of the different precodings in the \textit{Urban NLOS scenario}. Proposal algorithm is $\bm W_{ARZF} (\bm V)$. Optimal regularization $\bm W_\textit{OPT}(\bm V)$ was configured using the L-BFGS optimization algorithm.}
    \label{tab:Min SE Casual}
\end{table}

\begin{table}
    \centering
    \begin{tabular}{c | c}
        Parameter & Value \\
        \hline
        Base station parameters &  \\
        \hline
        number of base stations & $1$ \\
        position, m: (x, y, z) axes & $(0, 0, 25)$ \\
        number of antenna placeholders (y axis) & $8$ \\
        number of antenna placeholders (z axis) & $4$ \\
        distance between placeholders ($y$ axis) & $0.5$ wavelength \\
        distance between placeholders ($z$ axis) & $1.7$ wavelength \\
        antenna model & 3gpp-macro \\
        half-Power in azimuth direction, deg & 60 \\
        half-Power in elevation direction, deg & 10  \\
        front-to back ratio, dB  & 20 \\
        total number of antennas & 64 \\
        \hline
         Receiver parameters &  \\
        \hline
        number of placeholders at the receiver ($x$ axis) & $2$ \\
        distance between placeholders ($x$ axis) & $0.5$ wavelength \\
        antenna model & half-wave-dipole \\
        total number of antennas & 4 \\
        \hline
        Quadriga simulation parameters &  \\
        \hline
        central band frequency & $3.5$ GHz \\
        1 sample per meter (default value) & 1 \\
        include delay of the LOS path & 1 \\
        disable spherical waves (use\_3GPP\_baseline) & 1 \\
        \hline
        Quadriga channel builders parameters &  \\
        \hline
        shadow fading sigma & 0 \\
        cluster splitting & False \\
        bandwidth & $100$ MHz \\
        number of subcarriers & 42 \\ \\
    \end{tabular} 
    \caption{Quadriga generation parameters}
    \label{tab:quadriga_hyp}
\end{table}

\section*{Abbreviations}

\begin{tabular}{c|c}
    MIMO & Multiple-input multiple-output  \\
    ZF & Zero-Forcing \\
    RZF & Regularized Zero-Forcing \\
    MRT & Maximum Ratio Transmission \\
    DPC & Dirty Paper Coding \\
    UE & User equipment \\
    CEU & Cell Edge User equipment \\
    RSRP & Reference Signal Received Power \\
    VP & Vector Perturbation \\
    CD & Conjugate Detection \\
    SINR & Signal-to-Interference-and-Noise \\
    HARQ & Hybrid Automatic Repeat Request \\
    BLER & Block Error Rate \\ 
    SE & Spectral Efficiency \\
    MSE & Mean Squared Error \\
    MMSE & Minimum Mean Squared Error \\
    WRZF & Wiener Filter Zero-Forcing \\
    ARZF & Adaptive Regularized Zero-Forcing \\
    OPT & Gradient-Based Optimal Regularization \\
    LOS & Line-of-Sight \\
    NLOS & Non-Line-of-Sight \\
    SVD & Singular Value Decomposition \\
    CSI & Channel state information \\
    TDD & Time division duplex 
\end{tabular}

\bibliographystyle{unsrt}
\bibliography{interacttfssample}

\end{document}